\documentclass{lmcs}
\pdfoutput=1

\usepackage{lastpage}
\lmcsdoi{15}{1}{2}
\lmcsheading{}{\pageref{LastPage}}{}{}%
{Feb.~16,~2018}{Jan.~09,~2019}{}


\usepackage{amssymb}
\usepackage{hyperref}
\usepackage{stmaryrd}
\usepackage{tikz-cd}
\usetikzlibrary{decorations.pathmorphing,shapes}
\usepackage[all,pdf]{xy}

\theoremstyle{plain} 
\newcommand{\abs}[2]{\langle#1\tot#2\rangle}

\newcommand{\at}{\,}
\newcommand{\bij}{\cong}
\newcommand{\bimp}{\Leftrightarrow}
\newcommand{\Bool}{\mathtt{Bool}}
\newcommand{\C}{\mathbf{C}}

\newcommand{\code}[1]{\mathtt{#1}}
\newcommand{\comp}{\circ}
\renewcommand{\conj}{\wedge}

\newcommand{\Cov}{\mathit{J}}

\newcommand{\CwF}{\mathbf{C}}
\newcommand{\disj}{\vee}

\newcommand{\defeq}{\mathrel{{:}{\equiv}}}

\newcommand{\E}{\mathcal{E}}

\newcommand{\ent}{\vdash}
\newcommand{\false}{\mathtt{false}}

\newcommand{\Fib}{\mathbf{Fib}}
\newcommand{\from}[2]{#2_{{\totop}#1}}
\DeclareMathOperator{\fst}{\mathtt{fst}}
\newcommand{\fun}{\shortrightarrow}
\DeclareMathOperator{\funext}{\mathtt{funext}}

\DeclareMathOperator{\happly}{\mathtt{happly}}
\newcommand{\Ho}{\mathbf{Ho}}
\newcommand{\I}{\pos{\R}}
\newcommand{\id}{\mathtt{id}}

\DeclareMathOperator{\Id}{\mathtt{Id}}
\DeclareMathOperator{\idp}{\mathtt{idp}}

\newcommand{\imp}{\Rightarrow}
\newcommand{\inl}{\mathtt{inl}}
\newcommand{\inr}{\mathtt{inr}}

\newcommand{\iso}{\cong}
\DeclareMathOperator{\J}{\mathtt{J}}
\newcommand{\K}{\mathtt{K}}
\DeclareMathOperator{\lift}{\mathtt{lift}}
\renewcommand{\min}{\mathtt{min}}
\DeclareMathOperator{\minus}{\mathtt{-}}
\newcommand{\mono}{\rightarrowtail}
\providecommand{\monus}{
  \mathbin{
    \vphantom{+}
    \text{
      \normalfont
      \mathsurround=0pt 
      \ooalign{
        \noalign{\kern-.35ex}
        \hidewidth$\smash{\cdot}$\hidewidth\cr 
        \noalign{\kern.35ex}
        $-$\cr 
      }%
    }%
  }%
}
\newcommand{\morphism}{\longrightarrow}
\newcommand{\mult}{\cdot}
\newcommand{\N}{\mathbb{N}}

\newcommand{\op}{\mathrm{op}}

\newcommand{\pair}[2]{\langle#1\mathbin{,}#2\rangle}
\newcommand{\pth}{\sim}
\DeclareMathOperator{\Path}{\wp}
\newcommand{\pcomp}{\bullet}
\newcommand{\pcong}{\mathbin{\text{\normalfont\texttt{'}}}}

\newcommand{\pos}[1]{#1_{\text{\normalfont\texttt{+}}}}

\newcommand{\R}{\mathtt{R}}
\DeclareMathOperator{\refl}{\mathtt{Refl}}

\DeclareMathOperator{\rev}{\mathtt{rev}}
\newcommand{\RR}{\mathbb{R}}
\newcommand{\Set}{\mathbf{Set}}
\newcommand{\Sh}{\mathbf{Sh}}
\newcommand{\shape}[1]{\mathopen{\normalfont\text{\texttt{|}}}#1
  \mathclose{\normalfont\text{\texttt{|}}}}
\DeclareMathOperator{\snd}{\mathtt{snd}}
\newcommand{\source}{\partial^{\text{\normalfont\texttt{-}}}}
\newcommand{\src}{\mathtt{0}}

\DeclareMathOperator{\supp}{\mathtt{sup}}
\newcommand{\T}{\mathbf{T}}

\newcommand{\target}{\partial^{\text{\normalfont\texttt{+}}}}

\newcommand{\Top}{\mathbf{Haus}}
\newcommand{\tot}{\leqslant}
\newcommand{\totop}{\geqslant}
\newcommand{\trpt}[1]{#1_*}
\newcommand{\true}{\mathtt{true}}
\newcommand{\TT}{\mathtt{T}}
\newcommand{\UU}{\mathtt{U}}
\newcommand{\unit}{\mathtt{1}}

\newcommand{\upto}[2]{#2_{{\tot}#1}}
\DeclareMathOperator*{\W}{\mathtt{W}}

\newcommand{\Y}{\mathtt{Y}}

\begin{document}

\title[Models of Type Theory Based on Moore Paths] {Models of Type
  Theory Based on Moore Paths\rsuper*} \titlecomment{{\lsuper*}This is a
  revised and expanded version of a paper with the same name that
  appeared in the proceedings of the \emph{2nd International
    Conference on Formal Structures for Computation and Deduction}
  (FSCD 2017)}

\author[I.~Orton]{Ian Orton}
\address{University of Cambridge Department of Computer Science and Technology, UK}
\email{rio@cam.ac.uk}
\email{andrew.pitts@cl.cam.ac.uk}

\author[A.~M.~Pitts]{Andrew M. Pitts}	

\begin{abstract}
  This paper introduces a new family of models of intensional
  Martin-L\"of type theory. We use constructive ordered algebra in
  toposes. Identity types in the models are given by a notion of Moore
  path. By considering a particular gros topos, we show that there is
  such a model that is \emph{non-truncated}, i.e.~contains non-trivial
  structure at all dimensions. In other words, in this model a type in
  a nested sequence of identity types can contain more than one
  element, no matter how great the degree of nesting. Although
  inspired by existing non-truncated models of type theory based on
  simplicial and cubical sets, the notion of model presented here
  is notable for avoiding any form of Kan filling condition in the
  semantics of types.
\end{abstract}

\maketitle

\section{Introduction}
\label{sec:int}

Homotopy Type Theory~\cite{HoTT} has re-invigorated the study of the
intensional version of Martin-L\"of type
theory~\cite{Martin-LoefP:intttp}. On the one hand, the language of
type theory helps to express synthetic constructions and arguments in
homotopy theory and higher-dimensional category theory. On the other
hand, the geometric and algebraic insights of those branches of
mathematics shed new light on logical and type-theoretic notions. One
might say that the familiar \emph{propositions-as-types} analogy has
been extended to \emph{propositions-as-types-as-spaces}. In
particular, under this analogy there is a \emph{path-oriented} view of
intensional (i.e.~proof-relevant) equality: proofs of equality of two
elements $x,y$ of a type $A$ correspond to elements of a Martin-L\"of
identity type $\Id_A x\,y$ and behave analogously to paths between two
points $x,y$ in a space $A$. The complicated internal structure of
intensional identity types relates to the homotopy classes of path
spaces. To make the analogy precise and to exploit it, it helps to
have a wide range of models of intensional type theory that embody
this path-oriented view of equality in some way. This paper introduces
a new family of such models, constructed from \emph{Moore
  paths}~\cite{MooreJC:thefse} in toposes.

Let $\pos{\RR} = \{r\in\RR \mid r\geq 0\}$ be the real half-line with
the usual topology. Classically, a Moore path between points $x$ and
$y$ in a topological space $X$ is a pair $p = (f,r)$ where
$r\in\pos{\RR}$ and $f:\pos{\RR}\morphism X$ is a continuous function
with $f\,0 =x$ and $f\,r' = y$ for all $r'\geq r$. We will write
$x\pth y$ for the set of Moore paths from $x$ to $y$, with $X$
understood from the context. Clearly there is a Moore path from $x$ to
$y$ iff there is a conventional path, that is, a continuous function
$f:[0,1]\morphism X$ with $f\,0 = x$ and $f\,1=y$. The advantage of
Moore paths is that they admit degeneracy and composition operations
that are unitary and associative up to equality; whereas for
conventional paths these identities only hold up to homotopy.
Specifically, one has the following Moore paths:
\[
  \idp x \defeq (\lambda\, t \fun x,0)\in x\pth x
\]
and for all $p=(f,r)\in x\pth y$  and  $q=(g,s)\in y\pth z$  
\begin{align*}
    q\pcomp p &\defeq (g\pcomp^r f, r+s) \in x \pth z \\
  \text{where\ } (g\pcomp^r f)\,t &\defeq
  \begin{cases}
    f\,t &\text{if $t\leq r$}\\
    g(t - r) &\text{if $r\leq t$}
  \end{cases}
  \qquad(t\in\pos{\RR})
\end{align*}
These definitions satisfy $p\pcomp{\idp x} = p = {\idp y}\pcomp p$ and
$r \pcomp(q\pcomp p) = (r\pcomp q)\pcomp p$. In Section~\ref{sec:moopt}
we abstract from $\RR$ some simple
order-algebraic~\cite[chapter~VI]{BourbakiN:alg} structure sufficient
for the above definitions to work in a constructive algebraic setting,
rather than a classical topological one. Initially the structure of an
\emph{ordered abelian group} in some
topos~\cite{JohnstonePT:skeett,MacLaneS:shegl} suffices and then we
extend that to an \emph{ordered commutative ring} to ensure the models
satisfy function extensionality.

In Sections~\ref{sec:tap}--\ref{sec:uni} we use this structure in toposes
to develop a family of models of intensional Martin-L\"of type theory
with: identity types given by Moore paths, $\Sigma$-types, $\Pi$-types
satisfying function extensionality, inductive types (we just consider
disjoint unions and $W$-types) and Tarski-style universes. By
considering a particular \emph{gros} topos~\cite{SGA4} in
Section~\ref{sec:grotm} we get a \emph{non-truncated} instance of our
model construction, in other words one where iterated identity types
$\Id_A,\Id_{\Id_A},\Id_{\Id_{\Id_A}},\ldots$ can be non-trivial to any
depth of iteration.

The observation that the strictly associative and unitary nature of
composition of Moore paths aids in the interpretation of Martin-L\"of
identity types is not new; see for example \cite[Sections~5.1 and
5.5]{VanDenBergB:topsmi}. However, the fact that function
extensionality can hold for identity based on Moore paths
(Theorem~\ref{thm:funext}) is new and quite surprising, given that
such paths carry an intensional component, namely their ``shape''
(Definition~\ref{def:moopo}). Another novelty of our approach concerns
the fact that existing non-truncated models of type theory typically
make use of some form of Kan filling condition~\cite{KanDM:comdhg} to
define a class of \emph{fibrant} families of types with respect to
which path types behave as identity types. One of the contributions of
this paper is to show that one can avoid any form of Kan filling and
still get a non-truncated model of intensional Martin-L\"of type
theory.  Instead we use path composition and a simple notion of
fibrant family phrased just in terms of the usual operation of
transporting elements along equality proofs
(Definition~\ref{def:tapf}). As a consequence every type, regarded as
a family over the terminal type, is fibrant in our setting. In
particular, this means that intervals are first-class types in our
models, something which is not true for existing path-oriented models,
such as the classical simplicial~\cite{LumsdainePL:simmuf} and
constructive cubical
sets~\cite{CoquandT:modttc,CoquandT:cubttc,LicataDR:carctt} models;
and constructing universes in our setting does not need proofs of
their fibrancy.

\subsection*{Informal type theory} 

The new models of type theory we present are given in terms of
\emph{categories with families}
(CwF)~\cite{DybjerP:intt,HofmannM:synsdt}. Specifically, we start with
an arbitrary topos $\E$, to which can be associated a CwF, for example
as in~\cite{LumsdainePL:locumo,AwodeyS:natmht}. We write $\E(\Gamma)$
for the set of families indexed by an object $\Gamma\in\E$ and
$\E(\Gamma\ent A)$ for the set of elements of a family
$A\in\E(\Gamma)$. One can make $\E(\Gamma)$ into a category whose
morphisms between two families $A,B\in\E(\Gamma)$ are elements in
$\E(\Gamma.A\ent B)$, where $\Gamma.A$ is the comprehension object
associated with $A$. The category $\E(\Gamma)$ is equivalent to the
slice category $\E/\Gamma$, the equivalence being given on objects by
sending families $A\in\E(\Gamma)$ to corresponding projection
morphisms $\Gamma.A\morphism \Gamma$.

We then construct a new CwF by considering families in $\E$ equipped
with certain extra structure (the \emph{transport-along-paths}
structure of Definition~\ref{def:tapf}); the elements of a family in
the new CwF are just those of the underlying family in $\E$. One could
describe this construction using the language of category
theory. Instead, as in~\cite{PittsAM:aximct-jv,PittsAM:intumh} we find
it clearer to express the construction using an internal language for
the CwF associated with $\E$ -- a combination of higher-order
predicate logic and extensional type theory, with an impredicative
universe of propositions given by the subobject classifier in $\E$;
see~\cite{MaiettiME:modcdt}. This use of internal language allows us
to give an appealingly simple description of the type constructs in
the new CwF. In the text we use this language informally (analogously
to the way that \cite{HoTT} develops Homotopy Type Theory). In
particular the typing contexts of the judgements in the formal
version, such as $[x_0:A_0,x_1:A_1(x_0),x_2:A_2(x_0,x_1)]$, become
part of the running text in phrases like ``given $x_0:A_0$,
$x_1:A_1(x_0)$ and $x_2:A_2(x_0,x_1)$, then\ldots''; and when we refer
to a ``function in the topos'' (as opposed to one of its morphisms) we
mean a term of function type in its internal language.

The arguments we give in Sections~\ref{sec:moopt}--\ref{sec:uni} are
all constructively valid and in fact do not require the impredicative
aspects of topos theory; indeed we have used Agda~\cite{Agda} (with
uniqueness of identity proofs and postulates for quotient types) as a
tool to develop and experiment with the material presented in those
sections.  The specific model presented in Section~\ref{sec:grotm}
uses topological spaces within classical mathematics.

\section{Ordered Rings in  a Topos}
\label{sec:ordrt}

Let $\E$ be a topos with a natural number
object~\cite{JohnstonePT:skeett,MacLaneS:shegl}. A \emph{total order}
on an object $\R\in\E$ is given by a subobject ${\tot}\mono\R\times\R$
which is not only reflexive, transitive and anti-symmetric, but also
satisfies that the join of the subobject $\tot$ and its opposite
${\tot}\comp\pair{\pi_2}{\pi_1}$ is the whole of $\R\times\R$; in
other words the following formula of the internal language of $\E$ is
satisfied:
\begin{equation}
  \label{eq:1}
  (\forall i,j:\R)\; i\tot j \;\disj\; j\tot i
\end{equation}
As mentioned in the Introduction, in this paper we use such formulas
of the internal language to express properties of $\E$ instead of
giving the category-theoretic version of the property.  Note that
since $\E$ may not be Boolean, we do not necessarily have the
trichotomy property
$(\forall i,j:\R)\; i<j \;\disj\; i=j \;\disj\; j<i$ for the
associated strict order relation $i < j \;\defeq \neg(j\tot i)$. So
when defining functions by cases using \eqref{eq:1} we have to verify
that the clauses for $i\tot j$ and for $j\tot i$ agree on the overlap,
where $i=j$ holds by antisymmetry. For example, the \emph{positive
  cone}
\begin{equation}
  \label{eq:5}
  \I \defeq \{i:\R \mid \src\tot i\}
\end{equation}
associated with $\R$ has a binary operation of \emph{minimum}
$\min:\I\times\I\morphism\I$ well-defined by the following properties
\begin{equation}
  \label{eq:6}
  (\forall i,j:\I)
  \begin{array}[c]{l}
   i\tot j \;\imp\; \min(i, j) = i\\
   j\tot i \;\imp\; \min(i, j) = j
  \end{array}
\end{equation}  

\begin{figure}
  \emph{Total order}
  \begin{gather}
    (\forall i:\R)\;i \tot i\label{eq:38}\\
    (\forall i,j,k:\R)\; {i\tot j}\;\conj\; {j\tot k} \;\imp\; {i \tot k}\\
    (\forall i,j:\R)\; {i\tot j} \;\conj\; {j \tot i} \;\imp\; {i=j}\\
    (\forall i,j:\R)\; {i\tot j} \;\disj\; {j\tot i}\label{eq:2}
  \end{gather}
  \emph{Abelian group}
  \begin{gather}
    (\forall i,j,k:\R)\; i+(j+k) = (i+j)+k\label{eq:9}\\
    (\forall i:\R)\;\src+i = i\\
    (\forall i,j:\R)\;i+j = j+i\\
    (\forall i:\R)\; i+(-i) = \src 
  \end{gather}
  \emph{Addition is order preserving}
  \begin{equation}
    \label{eq:46}
    (\forall i,j,k:\R)\;{i\tot j} \;\imp\; {k+i}\tot{k+j}
  \end{equation}
  \emph{Multiplicative abelian monoid}
  \begin{gather}
    (\forall i,j,k:\R)\; i\mult(j\mult k) = (i\mult j)\mult k\label{eq:43}\\
    (\forall i:\R)\;\unit\mult i = i\\
    (\forall i,j:\R)\;i\mult j = j\mult i
  \end{gather}
  \emph{Multiplication distributes over addition}
  \begin{equation}
    \label{eq:37}
    (\forall i,j,k:\R)\;i\mult(j+k) = (i\mult j)+(i\mult k)
  \end{equation}
  \emph{Positive elements are closed under multiplication}
  \begin{equation}
    \label{eq:45}
    (\forall i,j:\R)\;{\src\tot i}\;\conj\;{\src\tot j} \;\imp\; {\src\tot
      i\mult j} 
  \end{equation}
  \caption{Ordered commutative ring axioms}
  \label{fig:ordcra}
\end{figure}

In order to define Moore paths in $\E$ with respect to $\R$ we need it
to have additive structure compatible with the total order; later, to
get function extensionality we also need multiplicative structure and
to construct universes we need a connectedness property for $\R$.

\begin{defi}[{order-ringed topos}]
  \label{def:ordrt}
  An \emph{order-ringed topos} $(\E,\R)$ is a topos $\E$ together with
  an \emph{ordered commutative ring
    object}~\cite[chapter~VI]{BourbakiN:alg} $\R$ in $\E$. Thus $\R$
  comes equipped with a subobject ${\tot}\mono\R\times\R$ and
  morphisms
  \[
    \src:1\morphism \R
    \qquad
    {\_+\_}:\R\times\R\morphism\R
    \qquad
    {-}:\R\morphism\R
    \qquad
    \unit:1\morphism\R
    \qquad
    {\_\mult\_}:\R\times\R\morphism\R
  \]
  satisfying the axioms in Fig.~\ref{fig:ordcra}.
\end{defi}

Specific examples of order-ringed toposes $(\E,\R)$ will be considered
in Section~\ref{sec:mod}. In the next section we develop properties of
Moore paths in $\E$ based on $\R$.

\section{Moore Paths in  a Topos}
\label{sec:moopt}

Fix an order-ringed topos $(\E,\R)$. Recall from the Introduction that
there is a CwF associated with
$\E$~\cite{LumsdainePL:locumo,AwodeyS:natmht} whose families at an
object $\Gamma\in\E$ we denote by $\E(\Gamma)$. We continue to use an
informal internal language based on extensional type theory to
describe constructions and properties of the topos and its associated
CwF. The development in this and the next section only makes use of
the order and additive structure of the positive cone
$\I=\{i:\R\mid \src\tot i\}$ of the ordered commutative ring object
$\R$.  (We will need to use its multiplicative structure in
Section~\ref{sec:fune}.)

\begin{defi}[{Moore path objects}]
  \label{def:moopo}
  For each object $\Gamma\in\E$, the family
  $(x\pth y \mid x,y:\Gamma)\in\E(\Gamma\times\Gamma)$ of \emph{Moore
    path objects} is defined by:
  \begin{equation}
    \label{eq:3}
    x\pth y \;\defeq\; \{(f,i):(\I\fun \Gamma)\times\I \mid f\,\src = x
    \;\conj\;(\forall j\totop i)\; f\,j = y\} 
  \end{equation}
  We have the following functions associated with Moore paths:
  \begin{align}
    &\shape{\_} : (x\pth y) \fun \I\label{eq:23}\\
    &\shape{(f,i)} = i\notag\\
    &\_\mathbin{\normalfont\text{\texttt{at}}}\_: (x \pth y) \fun \I
       \fun \Gamma\label{eq:4}\\  
    &(f,i)\mathbin{\normalfont\text{\texttt{at}}} j = f\,j\notag
  \end{align}
  Following~\cite[Section~2]{BrownR:moohsf}, we call $\shape{p}$ the
  \emph{shape} of the path $p:x\pth y$. Thus if $p:x\pth y$, then
  $p \mathbin{\normalfont\text{\texttt{at}}} \src = x$ and
  $p \mathbin{\normalfont\text{\texttt{at}}}\shape{p}=y$.  \emph{From
    now on we will just write
    $p \mathbin{\normalfont\text{\texttt{at}}}i$ as $p\at i$.}
\end{defi}

Note that morphisms in $\E$ respect Moore paths in the sense that for each
$\gamma:\Gamma\morphism \Delta$ in $\E$ there is a function
mapping paths in $\Gamma$ to those in $\Delta$:
\begin{align}
  &\gamma\pcong{} : (x \pth y) \fun (\gamma\,x \pth
    \gamma\,y) \label{eq:7}\\ 
  &\gamma\pcong(f,i) \defeq (\gamma\comp f,i)\notag
\end{align}

\begin{defi}[{Degenerate paths}]
  \label{def:degp}
  For each $\Gamma\in\E$, the \emph{degenerate path} at $x:\Gamma$ is
  denoted $\idp x: x\pth x$ and is well-defined by the requirements:
  \begin{gather}
    \shape{\idp x} = \src\\
    (\forall i:\I)\; (\idp x)\at i = x
  \end{gather}
\end{defi}

\begin{defi}[{Path composition}]
  Given $A\in\E$, if $i:\I$ and $f,g:\I\fun A$ satisfy $f\,i=g\,\src$,
  then there is a function $g\pcomp^i f : \I \fun A$ satisfying
  \[
    (g\pcomp^i f)\,j = 
    \begin{cases}
      f\,j &\text{if $j\tot i$}\\
      g\,(j - i) &\text{if $i \tot j$}
    \end{cases} 
  \]
  (where, as usual, we write $j+(\minus i)$ as $j-i$). This is
  well-defined because when $i=j$, then
  $f\,j = f\,i = g\,\src = g(j- i)$. In particular, given paths
  $p = (f,i):x\pth y$ and $q =(g,j):y\pth z$, then
  $g\pcomp^i f:\I\fun A$ satisfies
  $(\forall k\totop i+j)\;(g\pcomp^i f)\,k = z$, because $i+j \tot k$
  implies $i=i+\src\tot i+j \tot k$ and $j=(i+j)- i \tot k- i$; and
  hence $(g\pcomp^i f)\,k = g(k - i)=g\,j=z$. Therefore we get a
  well-defined path $q\pcomp p : x\pth z$ satisfying
  \begin{gather}
    \shape{q\pcomp p} = \shape{p}+\shape{q}\\
    (\forall i:\I)
    \begin{array}{rcl}
      i\tot\shape{p} &\imp& (q\pcomp p)\at i = p \at i\\
      \shape{p}\tot i &\imp& (q\pcomp p)\at i = q \at(i - \shape{p})
    \end{array}
  \end{gather}
\end{defi}

\begin{lem}
  \label{lem:degcp}
  For each $\Gamma\in\E$, given $x,y,z,w:\Gamma$, $p:x\pth y$, $q:y\pth z$, $r : z\pth w$ and $\gamma : \Gamma \fun \Gamma'$,
  one has:
  \begin{gather}
    p\pcomp(\idp x) = p = (\idp y)\pcomp p\\
    (r\pcomp q)\pcomp p = r\pcomp(q\pcomp p)\label{eq:8}\\
    \gamma\pcong(\idp x) = \idp(\gamma\,x)\\
    \gamma\pcong(q\pcomp p) = (\gamma\pcong q)\pcomp(\gamma\pcong p)
  \end{gather}
\end{lem}
\begin{proof}
  One just has to check that these properties follow constructively
  from the axioms \eqref{eq:38}--\eqref{eq:46}. For example
  \eqref{eq:8} holds because
  \[
    \shape{(r\pcomp q) \pcomp p} = \shape{p}+\shape{r\pcomp q} =
    \shape{p}+(\shape{q}+\shape{r}) = (\shape{p}+\shape{q})+\shape{r} =
    \shape{q\pcomp p} + \shape{r} = \shape{r\pcomp(q\pcomp p)}
  \]
  and
  \begin{align*}
    ((r\pcomp q) \pcomp p)\at i &=
    \begin{cases}
      p\at i &\text{if $i\tot\shape{p}$}\\
      q\at(i-\shape{p}) &\text{if $\shape{p}\tot i$ and $i-
        \shape{p}\tot \shape{q}$}\\
      r\at((i-\shape{p})-\shape{q}) &\text{if $\shape{p}\tot i$
        and $\shape{q}\tot i-\shape{p}$}
    \end{cases}\\
   (r\pcomp(q\pcomp p))\at i &=
      \begin{cases}
        p\at i &\text{if $i\tot\shape{p}+\shape{q}$ and
          $i\tot\shape{p}$}\\
        q\at(i-\shape{p}) &\text{if $i\tot\shape{p}+\shape{q}$
          and $\shape{p}\tot i$}\\
        r\at(i-(\shape{p}+\shape{q})) &\text{if
          $\shape{p}+\shape{q}\tot i$}
      \end{cases}
  \end{align*}
  which are equal, because one can use axioms
  \eqref{eq:38}--\eqref{eq:46} to show that
  \begin{align*}
    i\tot\shape{p} 
    &\;\bimp\; i\tot\shape{p}+\shape{q} \;\conj\; i\tot\shape{p}\\
    \shape{p}\tot i \;\conj\; i- \shape{p}\tot \shape{q} 
    &\;\bimp\; i\tot\shape{p}+\shape{q} \;\conj\; \shape{p}\tot i\\
    \shape{p}\tot i \;\conj\;\shape{q}\tot i-\shape{p}
    &\;\bimp\; \shape{p}+\shape{q}\tot i\\
    (i-\shape{p})-\shape{q}
    &\;=\; i-(\shape{p}+\shape{q}).
  \end{align*}
\end{proof}

As well as composing Moore paths one can reverse them. To define this
operation it is convenient to use the operation of \emph{truncated
  subtraction} ${\_\monus\_}:\I\times\I\morphism\I$, which is
well-defined by the following properties:
\begin{equation}
  \label{eq:11}
  (\forall i,j:\I)
  \begin{array}[c]{l}
    i\tot j \;\imp\; i\monus j = \src\\
    j\tot i \;\imp\; i\monus j = i-j
  \end{array}
\end{equation}

\begin{lem}[{path reversal}]
  \label{lem:patr}
  For each $\Gamma\in\E$, given $x,y:\Gamma$ and $p:x\pth y$, there
  is a \emph{reversed path} $\rev p : y\pth x$ well-defined by the
  requirements
  \begin{gather*}
    \shape{\rev p} = \shape{p}\\
    (\forall i:\I)\; (\rev p)\at i = p\at(\shape{p} \monus i)
  \end{gather*}
  and satisfying
  \begin{gather}
    \rev(\idp x) = \idp x\\
    \rev(q\pcomp p) = (\rev p)\pcomp(\rev q)\\
    \rev(\rev p ) = p\\
    \gamma\pcong(\rev p)= \rev(\gamma\pcong p)
  \end{gather}
\end{lem}
\begin{proof}
  As for the previous lemma, this follows straightforwardly from the
  axioms \eqref{eq:38}--\eqref{eq:46} within constructive logic.
\end{proof}

Although the definition of $\rev p$ is standard, the above equational
properties are not often mentioned in the literature. However, they
are crucial for the construction of identity types in
Section~\ref{sec:tap} to work. What usually gets a mention is the fact
that up to homotopy $\rev p$ is a two-sided inverse for $p$ with
respect to the $\pcomp$ operation. Paths $(\rev p)\pcomp p \pth \idp x$ and
$p\pcomp(\rev p)\pth \idp y$ can be constructed using the path
contraction operation given below; we do not bother to do that,
because they are also a consequence of the \emph{path
  induction}~\cite[1.12.1]{HoTT} property of identity types that
follows from Theorem~\ref{thm:idet}.

\begin{defi}[{bounded abstractions}]
  \label{def:boua}
  The following binding syntax is very convenient for describing Moore
  paths. For each $\Gamma\in\E$, if $\lambda i\fun \varphi(i)$ describes a
  function in $\I\fun \Gamma$, then for each $j:\I$ using the $\min$
  function~\eqref{eq:6} we get a path
  $\abs{i}{j}\varphi(i) : \varphi(\src) \pth \varphi(j)$ in $\Gamma$
  by defining:
  \begin{equation}
    \label{eq:31}
    \abs{i}{j}\varphi(i) \defeq (\lambda i\fun
    \varphi(\min(i,j))\mathrel{,} j)
  \end{equation}
  ($i$ is bound in the above expression). It is easy to see that this
  form of bounded abstraction has the following properties:
  \begin{align}
    \gamma\pcong \abs{i}{j}\varphi(i)
    &= \abs{i}{j}\gamma(\varphi(i))\label{eq:19}\\ 
    \abs{i}{\src}\varphi(i) &= \idp(\varphi(\src))\label{eq:20}\\
    \abs{i}{\shape{p}}(p\at i)  &= p \label{eq:21}
  \end{align}
\end{defi}

\begin{lem}[{path contraction}]
  \label{lem:patc}
  Given $\Gamma\in\E$, for any path $p:x\pth y$ in $\Gamma$ and
  $i:\I$, there is a path $\upto{i}{p} : x\pth p\at i$ satisfying
  \begin{gather}
    \upto{\src}{p} =\idp x\label{eq:40}\\
    (\forall i\totop\shape{p})\;\upto{i}{p} = p\label{eq:41}
  \end{gather}
\end{lem}
\begin{proof}
  Using the bounded abstraction notation and the $\min$
  function~\eqref{eq:6}, we define 
  \begin{equation}
    \label{eq:36}
    \upto{i}{p} \defeq \abs{j}{\min(\shape{p},i)}(p\at j)
  \end{equation}
  Since $p\at(\min(\shape{p},i)) = p\at i$, this does give a path
  $x\pth p\at i$; and it has the required properties by
  \eqref{eq:20} and \eqref{eq:21}.
\end{proof}

\begin{rem}
  \label{rem:sing-contr}
  Using $\upto{i}{p}$ we get for each $x:\Gamma$ that
  $\sum_{y:\Gamma}(x\pth y)$ is path-contractible
  ~\cite[Section~3.11]{HoTT} with centre $(x,\idp x)$, since for each
  $y:\Gamma$ and $p:x\pth y$ we have a path
  $\abs{i}{\shape{p}}(p\at i, \upto{i}{p})$ in
  $(x,\idp x) \pth (y,p)$.  This is part of the more general fact that
  Moore paths model identity types, which we show in the next section
  (see Theorem~\ref{thm:idet}).
\end{rem}

\section{Transport along Paths}
\label{sec:tap}

In this section we continue with the assumptions of the previous one:
$\E$ is a topos (with associated CwF) containing an ordered
commutative ring object $\R$. We want objects of Moore paths with
respect to $\R$ (Definition~\ref{def:moopo}) to give a model of
identity types, as well as other type formers. Recall that in
Martin-L\"of Type Theory elements of identity types $\Id_\Gamma x\,y$
give rise to transport functions $A\,x\fun A\,y$ between members of a
family of types $(A\,x\mid x:\Gamma)$ over a type $\Gamma$; see for
example~\cite[Section~2.3]{HoTT}. Therefore, for each object
$\Gamma\in\E$, we should restrict attention to families
$A\in\E(\Gamma)$ that come equipped at least with some sort of
transport operation taking a path $p:x\pth y$ in $\Gamma$ and an
element $a:A\,x$ to an element $\trpt{p}a : A\,y$. This leads to the
following definition.

\begin{defi}[{tap fibrations}]
  \label{def:tapf}
  Given an object $\Gamma\in\E$, a \emph{transport-along-paths} (tap)
  structure for a family $A\in\E(\Gamma)$ is a
  $(\Gamma\times\Gamma)$-indexed family of morphisms
  $(\trpt{(\_)} : x\pth y \fun (A\,x \fun A\,y) \mid x,y:\Gamma)$
  satisfying for all $x:\Gamma$ and $a:A\,x$
  \begin{equation}
    \label{eq:10}
    \trpt{(\idp x)}a = a 
  \end{equation}
  We write $\Fib(\Gamma)$ for the families over $\Gamma$ equipped with
  a tap structure and call them \emph{fibrations}. They are stable under
  re-indexing: given $\gamma:\Delta\morphism \Gamma$ in $\E$ and
  $A\in\Fib(\Gamma)$, then
  $A\,\gamma \defeq (A(\gamma\,x) \mid x:\Delta)\in \E(\Delta)$ has a
  tap structure via the congruence operation \eqref{eq:7}, taking the
  transport of $a:(A\,\gamma)\,x = A(\gamma\,x)$ along $p:x\pth y$ to
  be $\trpt{(\gamma\pcong p)} a : A(\gamma\,y)$. This re-indexing of tap
  structure respects composition and identities in $\E$. So $\Fib$
  inherits the structure of a CwF from that of $\E$, with the set of
  elements of a fibration being the elements of the underlying family
  in $\E$, that is $\Fib(\Gamma\ent A)=\E(\Gamma\ent A)$.
\end{defi}

We show that the CwF $\Fib$ inherits some type structure from $\E$. To
do so involves definitions and calculations using the bounded
abstraction formalism of Definition~\ref{def:boua}.

To describe $\Sigma$-{} and $\Pi$-types in
$\Fib$ we have to lift paths in $\Gamma$ to paths in comprehension
objects $\Gamma.A=\sum_{x:\Gamma} A\,x$ of fibrations
$A\in\Fib(\Gamma)$:

\begin{lem}[{path lifting}]
  \label{lem:patl}
  Given $\Gamma\in\E$ and $A\in\Fib(\Gamma)$, for each path
  $p: x\pth y$ in $\Gamma$ and each $a:A\,x$, there is a path
  $\lift(p,a) : (x,a)\pth (y, \trpt{p}a)$ in $\Gamma.A$
  satisfying
  \begin{equation} 
    \label{eq:14}
    \lift(\idp x , a) = \idp(x,a)
  \end{equation}
  and stable under re-indexing along morphisms $\Delta\morphism\Gamma$
  in $\E$.
\end{lem}
\begin{proof}
  We can use path contraction (Lemma~\ref{lem:patc}) to express path
  lifting using the bounded abstraction notation from
  Definition~\ref{def:boua}:
  \begin{equation}
    \label{eq:12}
    \lift(p,a) \defeq \abs{i}{\shape{p}}(p\at i,
    \trpt{(\upto{i}{p})}a)
  \end{equation}
  Since the path contraction operation satisfies
  $\upto{\src}{p}= \idp x$ and $\upto{\shape{p}}{p} = p$, this does
  indeed give a path in $(x,a)\pth (y, \trpt{p}a)$ . The desired
  properties of $\lift$ follow from corresponding properties of
  bounded abstraction and path
  contraction \eqref{eq:19}--\eqref{eq:41}. Specifically, we have:
  \begin{align*}
  	\lift(\idp x,a)
     &= \abs{i}{\shape{\idp x}}(\idp x \at i,\trpt{(\upto{i}{(\idp x)})}a) \\
     &= \abs{i}{\shape{\src}}(x,\trpt{(\upto{i}{(\idp x)})}a) \\
     &= \idp(x,\trpt{(\upto{\src}{(\idp x)})}a) \\
     &= \idp(x,\trpt{(\idp x)}a) \\
     &= \idp(x,a) 
  \end{align*}
  where $x :\Gamma$ and $a : A\,x$. We also have:
  \begin{align*}
  	\lift_A(\gamma \pcong p,a)
     &= \abs{i}{\shape{\gamma \pcong p}}((\gamma \pcong p)\at i,\trpt{(\upto{i}{(\gamma \pcong p)})}a) \\
     &= \abs{i}{\shape{p}}(\gamma (p\at i),\trpt{(\gamma \pcong \upto{i}{p})}a) \\
     &= \abs{i}{\shape{p}}(\gamma \times \id)(p\at i,\trpt{(\gamma \pcong \upto{i}{p})}a) \\
     &= (\gamma \times \id)\pcong \left(\abs{i}{\shape{p}}(p\at i,\trpt{(\gamma \pcong \upto{i}{p})}a)\right) \\
     &= (\gamma \times \id)\pcong (\lift_{A\gamma}(p,a))
  \end{align*}
  where $\gamma : \Delta \morphism \Gamma$, $p : x \pth y$ in $\Delta$, $a : A(\gamma\,x)$ and we write
  $\lift_A$ for the lifting operation on $A \in \Fib(\Gamma)$ and $\lift_{A\gamma}$ for the lifting operation
  on $A\gamma \in \Fib(\Delta)$.
\end{proof}

\begin{thm}[{$\Sigma$-{} and $\Pi$-types}]
  \label{thm:sptyp}
  Given $\Gamma\in\E$, $A\in \Fib(\Gamma)$ and $B\in\Fib(\Gamma.A)$,
  the families
  $\Sigma\,A\,B \defeq (\sum_{a:A\,x}B(x,a) \mid x:\Gamma)$ and
  $\Pi\,A\,B \defeq (\prod_{a:A\,x}B(x,a) \mid x:\Gamma)$ in
  $\E(\Gamma)$ have tap fibration structures that are stable under
  re-indexing along morphisms $\Delta\morphism\Gamma$ in $\E$. Hence
  the CwF $\Fib$ supports $\Sigma$-{}and
  $\Pi$-types~\cite[Definitions~3.15 and~3.18]{HofmannM:synsdt}.
\end{thm}
\begin{proof}
  Given a path $p:x\pth y$ in $\Gamma$, we get functions
  $\Sigma\,A\,B\,x\fun\Sigma\,A\,B\,y$ and
  $\Pi\,A\,B\,x\fun\Pi\,A\,B\,y$ using path lifting and (in the second
  case) path reversal: 
  \begin{align}
    \trpt{p}(a,b) 
    &\defeq (\trpt{p}a, \trpt{\lift(p,a)}b) \in \Sigma\,A\,B\,y
    &&\text{where $(a,b): \Sigma\,A\,B\,x$} \label{eq:16}\\
    (\trpt{p}f)\,a 
    &\defeq \trpt{(\rev(\lift(\rev p,a)))} f(\trpt{(\rev p)}a) \in
      B(y,a) 
    &&\begin{array}[c]{@{}l}
        \text{where $f: \Pi\,A\,B\,x$}\\[-\jot]
        \text{and $a:A\,y$}\
      \end{array}\label{eq:17}
  \end{align}
  To see why \eqref{eq:17} has the correct type, consider an arbitrary
  $p:x\pth y$ in $\Gamma$, $f: \Pi\,A\,B\,x$ and $a:A\,y$. We have
  $(\rev p) : y \pth x$, and hence $\trpt{(\rev p)}a : A\,x$ and
  $f(\trpt{(\rev p)}a) : B(x,\trpt{(\rev p)}a)$.
  Next, we have
  $\lift(\rev p,a) : (y,a) \pth (x,\trpt{\rev p}a)$ and therefore
    \[\rev(\lift(\rev p,a)) : (x,\trpt{\rev p}a) \pth (y,a)\]
  This allows us to transport along this path to get
    \[ \trpt{(\rev(\lift(\rev p,a)))} f(\trpt{(\rev p)}a) \in  B(y,a) \]
  as required.
  
  Note that these definitions satisfy the required property when transporting
  along the identity path. That is, given $x : \Gamma$ and $(a,b) : \Sigma\,A\,B\,x$,
  then using property \eqref{eq:14} we have:
    \[ \trpt{(\idp x)}(a,b)
       = (\trpt{(\idp x)}a, \trpt{\lift((\idp x),a)}b)
       = (\trpt{(\idp x)}a, \trpt{(\idp (x,a))}b)
       = (a, b) \]
  Similarly, given $x : \Gamma$, $f: \Pi\,A\,B\,x$ and $a:A\,x$ then using \eqref{eq:14}
  and the properties of path reversal given in Lemma~\ref{lem:patr} we have:
  \begin{align*}
  	(\trpt{(\idp x)}f)\,a
  	  &= \trpt{(\rev(\lift(\rev (\idp x),a)))} f(\trpt{(\rev (\idp x))}a) \\ 
  	  &= \trpt{(\rev(\lift(\idp x,a)))} f(\trpt{(\idp x)}a) \\  
  	  &= \trpt{(\rev(\idp(x,a)))} f(a) \\ 
  	  &= \trpt{(\idp(x,a))} f(a) \\ 
  	  &= f(a)
  \end{align*}
  Finally, the stability of these definitions under re-indexing comes from the
  fact that $\rev$ and $\lift$ are both stable under re-indexing.
\end{proof}

\begin{thm}[{Empty, unit, Boolean and natural number types}]
  \label{thm:empub}
  Given $\Gamma\in\E$, for each $A\in\E$, the constant family
  $(A\mid x:\Gamma)$ has a tap fibration structure given by
  $\trpt{p}a=a$ (for any $a:A$) and this is stable under
  re-indexing. Taking $A$ to be the initial object $\emptyset$, the
  terminal object $1$, the coproduct $1+1$ and the natural number
  object of the topos, we have that the CwF $\Fib$ supports empty,
  unit, Boolean and natural number
  types~\cite[Exercises~E3.24--E3.26]{HofmannM:synsdt}. \qed
\end{thm}

\begin{thm}[{Sum types}]
  \label{thm:sumt}
  Given $\Gamma\in\E$ and $A,B\in\Fib(\Gamma)$, the family
  $A\oplus B \defeq (A\,x + B\,x \mid x:\Gamma)$ of sum types has a
  tap fibration structure that is stable under re-indexing along
  morphisms $\Delta\morphism\Gamma$ in $\E$. Hence the CwF $\Fib$
  supports sum types.
\end{thm}
\begin{proof}
  Given a path $p:x\pth y$ in $\Gamma$, we get a function $(A\oplus
  B)\,x \morphism (A\oplus B)\,y$ by case analysis on the elements of
  $(A\oplus B)\,x = A\,x + B\,x$. Thus if $\inl$ and $\inr$ denote the
  constructors of the sum type $A\,x + B\,x$, we have
  \begin{align}
    \trpt{p}(\inl\,a) &= \inl(\trpt{p}a) &&\text{where $a:A\,x$}\\
    \trpt{p}(\inr\,b) &= \inr(\trpt{p}b) &&\text{where $b:B\,x$}
  \end{align}
  and clearly this definition inherits property \eqref{eq:10} from the
  tap fibration structure of $A$ and $B$, and is stable under
  re-indexing.
\end{proof}

Since we only consider toposes with a natural number object, the CwF
associated with $\E$ can interpret types of well-founded trees
($W$-types)~\cite[Section~5.3]{HoTT}; see \cite[Propositions~3.6 and
3.8]{MoerdijkI:weltc}. We write $\W_{x:A}B\,x$ for the object of
well-founded trees determined by a family $B\in\E(A)$, with
constructor
$\supp:\sum_{y:A}(B\,y\fun \W_{x:A}B\,x) \fun \W_{x:A}B\,x$.

\begin{thm}[{$W$-types}]
  \label{thm:wtyp}
  Given $A\in\Fib(\Gamma)$ and $B\in\Fib(\Gamma.A)$,
  the family
  \[\textstyle
  W A\,B \defeq (\W_{a:A\,x}B(x,a) \mid x:\Gamma) \in \E(\Gamma)
  \]
  has a tap fibration structure that is stable under re-indexing along
  morphisms $\Delta\morphism\Gamma$ in $\E$. Therefore the CwF $\Fib$
  supports $W$-types.
\end{thm}
\begin{proof}
  Given a path  $p:x\pth y$ in $\Gamma$, we get a function $W A\,B\,x
  \fun WA\,B\,y$ via the following well-founded recursion equation:
  \begin{equation}
    \label{eq:22}
    \trpt{p}\supp(a\mathbin{,}f) = \supp\left(\trpt{p}a \mathbin{,}
      \lambda b \fun \trpt{p}f(\trpt{(\rev(\lift(p,a)))}b)\right)
    \quad
    \begin{array}[c]{@{}l}
      \text{where $a:A\,x$ and}\\
      f:B(x,a)\fun W A\,B\,x
    \end{array}   
  \end{equation}
  Note that
  $\trpt{p} a : A\,y$. Therefore, assuming that the second argument
  to the $\supp$ constructor has type $B(y,\trpt{p} a) \fun W A\,B\,y$,
  then the overall type of the constructor will be $W A\,B\,y$ as required.
  To see why the second component does have this type, consider
  an arbitrary $b : B(y,\trpt{p} a)$ and observe that
  $\rev(\lift(p,a)) : (y,\trpt{p}a) \pth (x,a)$. Therefore we have
  $\trpt{(\rev(\lift(p,a)))}b : B(x,a)$ and hence
  $f(\trpt{(\rev(\lift(p,a)))}b) : W A\,B\,x$. Finally,
  recursively transporting along $p$ gives us
  $\trpt{p}f(\trpt{(\rev(\lift(p,a)))}b) : W A\, B\, y$ as required.
  
  Well-founded inductions using the properties of reversal and lifting
  given in Lemmas~\ref{lem:patr} and \ref{lem:patl} suffice to show
  that this inherits the properties of a tap structure from those for
  $A$ and $B$:
  \begin{align*}
  	\trpt{(\idp x)}\supp(a\mathbin{,}f)
  	&= \supp\left(\trpt{(\idp x)}a \mathbin{,}
      \lambda b \fun \trpt{(\idp x)}f(\trpt{(\rev(\lift(\idp x,a)))}b)\right) \\
  	&= \supp\left(\trpt{(\idp x)}a \mathbin{,}
      \lambda b \fun \trpt{(\idp x)}f(\trpt{(\idp(x,a))}b)\right) \\
  	&= \supp\left(a \mathbin{,}
      \lambda b \fun f(b)\right) \\
  	&= \supp(a \mathbin{,} f)
  \end{align*}
  The fact that this definition is stable under re-indexing follows
  immediately from the fact that $\rev$ and $\lift$ are both stable
  under re-indexing.
\end{proof}

So far we have considered type structure that lifts from the CwF
associated with $\E$ to the CwF $\Fib$. Now we consider identity
types, where the structure of interest in $\Fib$ is not the one
inherited from $\E$. Since $\E$ is a topos, it certainly has
\emph{extensional} identity types~\cite{Martin-LoefP:inttt},
inhabitation of which coincides with judgemental equality, and those
could be lifted to $\Fib$. However, we wish to show that Moore path
objects give the \emph{intensional} version of identity types in
$\Fib$, the family of types $(\Id_A x\,y \mid x,y:A)$ inductively
generated by a single constructor $\refl_A : \prod_{x:A} \Id_A
x\,x$. We will use Hofmann's version of the structure in a CwF needed
to model such types~\cite{HofmannM:synsdt}. To do so, let us fix some
notation for a CwF $\CwF$.

Re-indexing of a family $A\in\CwF(\Gamma)$ and an element
$\alpha\in\CwF(\Gamma\ent A)$ along a morphism
$\gamma:\Delta\morphism\Gamma$ will just be denoted by
$A\,\gamma\in\CwF(\Delta)$ and
$\alpha\,\gamma\in \CwF(\Delta\ent A\,\gamma)$; and given an element
$\beta\in\CwF(\Delta\ent A\,\gamma)$, then $\pair{\gamma}{\beta}$
denotes the unique morphism $\Delta\morphism \Gamma.A$ whose
composition with ${\fst}:\Gamma.A\morphism \Gamma$ is $\gamma$ and
whose re-indexing of the generic element
${\snd}\in \CwF(\Gamma.A\ent A\,{\fst})$ is $\beta$:
\[
  \fst\comp \pair{\gamma}{\beta} = \gamma
  \qquad
  \snd \pair{\gamma}{\beta} = \beta.
\]

\begin{defi}
  \label{def:idet}
  Following Hofmann~\cite[Definition~3.19]{HofmannM:synsdt}, we say
  that a CwF $\CwF$ \emph{supports the interpretation of intensional
    identity types} if for each object $\Gamma\in\CwF$ and each family
  $A\in\CwF(\Gamma)$ the following data is given and is stable under
  re-indexing along any $\gamma:\Delta\morphism \Gamma$:
  \begin{itemize}

  \item a family $\Id_A\in\C(\Gamma.A.A\,{\fst})$,

  \item a morphism
    $\refl_A : \Gamma.A\morphism \Gamma.A.A\,{\fst}.\Id_A$ such that
    ${\fst}\comp\refl_A$ equals the diagonal morphism
    $\pair{\id}{\snd}: \Gamma.A\morphism
    \Gamma.A.A\,{\fst}$,

  \item a function mapping each $B\in\C(\Gamma.A.A\,{\fst}.\Id_A)$ and
    $\beta\in\C(\Gamma.A \ent B\,\refl_A)$ to an element $\J_A
    B\,\beta\in\C(\Gamma.A.A\,{\fst}.\Id_A \ent B)$ such that the
    re-indexing $(\J_A B\,\beta)\,\refl_A$ equals $\beta$.
  \end{itemize}
\end{defi}

Given an object $\Gamma$ of the topos $\E$ and a family
$A\in\E(\Gamma)$, we can use the family of Moore path objects
$\_\pth\_$ (Definition~\ref{def:moopo}) to define a family
$\Id_A\in\E(\Gamma.A. A\,{\fst})$ as follows:
\begin{equation}
  \label{eq:18}
  \Id_A \defeq (a_1 \pth a_2 \mid ((x,a_1), a_2):\Gamma.A.A\,{\fst})  
\end{equation}
We will show that this together with suitable $\refl$ and $\J$
operations give an instance of Definition~\ref{def:idet} for the CwF
$\Fib$. In particular $\Id_A$ has a tap fibration structure when $A$
does. To see this we first need to analyse paths in $\Gamma.A$ in
terms of paths in $\Gamma$ and in the fibres $A\,x$ (for $x:\Gamma$).

\begin{lem}
  \label{lem:star}
  Given $\Gamma\in\E$ and $A\in\Fib(\Gamma)$, for each path
  $p:(x,a)\pth (y,b)$ in $\Gamma.A$, there is a path 
  $\snd(p) : \trpt{({\fst}\pcong p)} a \pth b$ in $A\,y$
  satisfying
  \begin{equation}
    \label{eq:33}
    \snd(\idp(x,a)) = \idp a
  \end{equation}
  and stable under re-indexing along any
  $\gamma:\Delta\morphism \Gamma$ in $\E$.
\end{lem}
\begin{proof}
  We use a reversed version of the path-contraction operation from
  Lemma~\ref{lem:patc}: for each path $q:x\pth y$ and each $i:\I$,
  define $\from{i}{q}:q\at i \pth y$ by
  \begin{equation}
    \label{eq:32}
    \from{i}{q} \defeq \rev\left(\upto{(\shape{q}\monus i)}{(\rev
        q)}\right) 
  \end{equation}
  Given $p:(x,a)\pth (y,b)$, we get
  ${\fst}\pcong p:x\pth y$ and hence for each $i:\I$ we have
  $\from{i}{({\fst}\pcong p)} : ({\fst}\pcong p)\at i \pth y$; and
  since
  $\snd(p\at i) : A(\fst(p\at i)) = A(({\fst}\pcong p)\at i)$, we
  get $\trpt{(\from{i}{({\fst}\pcong p)})}\snd(p\at i) : A\,y$. So we
  can define
  \begin{gather}
    \label{eq:79}
    \snd(p) = \abs{i}{\shape{p}}\left(\trpt{(\from{i}{({\fst}\pcong
        p)})}\snd(p\at i)\right) 
  \end{gather}
  to get a path $\trpt{({\fst}\pcong p)}a \pth b$ in
  $A\,y$. Property \eqref{eq:33} holds since, using \eqref{eq:20}, we have
  \begin{align*}
  	\snd(\idp(x,a)) 
  	  &= \abs{i}{\shape{\idp(x,a)}}\left(\trpt{(\from{i}{({\fst}\pcong
        \idp(x,a))})}\snd(\idp(x,a)\at i)\right)\\
      &= \abs{i}{\src}\left(\trpt{(\from{i}{
        (\idp x)})}\snd(x,a)\right)\\
      &= \idp\left(\trpt{(\from{\src}{(\idp x)})} a\right)\\
      &= \idp\left(\trpt{(\idp x)} a\right)\\
      &= \idp a
  \end{align*} 
  Stability under re-indexing follows from \eqref{eq:19} and the fact
  that $\rev$ is stable under re-indexing.
\end{proof}

\begin{lem}
  \label{lem:idet}
  Given $\Gamma\in\E$ and a fibration $A\in\Fib(\Gamma)$, the family
  $\Id_A\in\E(\Gamma.A.A\,{\fst})$ defined in \eqref{eq:18} has a tap
  fibration structure that is stable under re-indexing along morphisms
  $\Delta\morphism \Gamma$ in $\E$.
\end{lem}
\begin{proof}
  First note that re-indexing the fibration $A\in\Fib(\Gamma)$ along
  ${\fst}:\Gamma.A\morphism \Gamma$ we get
  $A\,{\fst} \in\Fib(\Gamma.A)$. Given a path
  $p:((x,a_1),a_2)\pth((y,b_1),b_2)$ in $\Gamma.A.A\,{\fst}$, define
  \[
    p' \defeq {\fst}\pcong({\fst}\pcong p) : x \pth y
  \]
  We can apply Lemma~\ref{lem:star} to the paths
  ${\fst}\pcong p : (x,a_1)\pth (y,b_1)$ and
  $\pair{\fst\comp\fst}{\snd}\pcong p : (x,a_2)\pth (y,b_2)$ in
  $\Gamma.A$ to get paths in $A\,y$
  \begin{align*}
    p_1 &\defeq\snd({\fst}\pcong p): \trpt{{p'}}a_1 \pth b_1\\
    p_2 &\defeq \snd(\pair{\fst\comp\fst}{\snd}\pcong p) :
  \trpt{{p'}}a_2 \pth b_2
  \end{align*}
  (using the fact that
  $\fst\comp \pair{\fst\comp\fst}{\snd} = \fst\comp\fst$). Thus for
  each path $q:a_1\pth a_2$ in $A\,x$, we have
  \[
    \begin{array}{ccc}
      A\,x
      & \xrightarrow{\trpt{{p'}}}
      & A\,y\\
      \null\\
      \begin{tikzcd} a_1 \arrow[r,squiggly,"q"] & a_2 \end{tikzcd} 
      &&
        \begin{tikzcd}[row sep = large]
            \trpt{{p'}}a_1 \arrow[r, squiggly, "(\trpt{{p'}})\pcong q"]
          & \trpt{{p'}}a_2 \arrow[d, squiggly, "p_2"] \\
            b_1 \arrow[u, squiggly, "\rev(p_1)"]
          & b_2
        \end{tikzcd}                                      
    \end{array}
    \vspace{7pt}
  \]
  and can compose together the paths in $A\,y$ to get
  $p_2\pcomp ((\trpt{{p'}})\pcong q) \pcomp \rev(p_1) : b_1\pth b_2$.
  So altogether we get a function
  $\trpt{p}: \Id_A((x,a_1),a_2) \fun \Id_A((y,b_1),b_2)$ defined by:
  \begin{equation}
    \label{eq:48}
    \trpt{p}q = \snd(\pair{\fst\comp\fst}{\snd}\pcong
    p)\pcomp\left((\trpt{({\fst}\pcong({\fst}\pcong 
        p))})\pcong q\right) \pcomp \rev(\snd({\fst}\pcong p))
  \end{equation}
  The properties of Moore path reversal (Lemma~\ref{lem:patr})
  together with Lemma~\ref{lem:star} suffice to show that this
  definition inherits the property of a tap fibration structure
  \eqref{eq:10} from the one for $A$ and that it is stable under
  re-indexing. For example, given $x : \Gamma$, $a_1, a_2 : A\,x$
  and $q : a_1 \pth a_2$ we have:
  \begin{align*}
  	\trpt{(\idp((x,a_1),a_2))}q
  	  &= \snd(\pair{\fst\comp\fst}{\snd}\pcong\idp((x,a_1),a_2)) \\
      &\qquad \pcomp\left((\trpt{({\fst}\pcong({\fst}\pcong 
        \idp((x,a_1),a_2)))})\pcong q\right) \\
      &\qquad \pcomp \rev(\snd({\fst}\pcong \idp((x,a_1),a_2)))\\
  	  &= \snd(\idp(x,a_2)) \pcomp\left((\trpt{(\idp x)})\pcong q\right) \pcomp \rev(\snd(\idp(x,a_1)))\\
  	  &= (\idp a_2) \pcomp q \pcomp (\idp a_1)\\
  	  &= q \\
  \end{align*}
  as required.
\end{proof}

\begin{thm}[{Identity types}]
  \label{thm:idet}
  The CwF $\Fib$ of Definition~\ref{def:tapf} supports the
  interpretation of intensional identity types
  (Definition~\ref{def:idet}), given by Moore path objects as in
  \eqref{eq:18}.
\end{thm}
\begin{proof}
  In view of Lemma~\ref{lem:idet}, it just remains to define the
  $\refl$ and $\J$ operations as in Definition~\ref{def:idet}.  Given
  $\Gamma\in\E$ and $A\in\Fib(\Gamma)$, we get $\refl_A :
  \Gamma.A\morphism \Gamma.A.A\,{\fst}.\Id_A$ by defining
  \begin{equation}
    \label{eq:35}
    \refl_A(x,a) \defeq (((x,a),a),\idp a)
  \end{equation}
  Note that
  $({\fst}\comp\refl_A)(x,a) = ((x,a),a) = \pair{\id}{\snd}(x,a)$, as
  required.

  To define $\J$ we combine transport along paths with the fact that
  singleton types are contractible (Remark~\ref{rem:sing-contr}). More
  specifically, given $B\in\Fib(\Gamma.A.A\,{\fst}.\Id_A)$ and
  $\beta\in\E(\Gamma.A \ent B\,\refl_A)$, for each path
  $p:a_1\pth a_2$ in $A\,x$ using Lemma~\ref{lem:patc} we have the
  following path in $\Gamma.A.A\,{\fst}.\Id_A$:
  \[
  \abs{i}{\shape{p}}(((x,a_1), p\at i), \upto{i}{p}) :
  (((x,a_1),a_1),\idp a_1) \pth (((x,a_1),a_2),p)
  \]
  Since $B$ has a tap fibration structure we can transport
  $\beta(x,a_1):B(((x,a_1),a_1),\idp a_1)$ along this path to get an
  element of $B(((x,a_1),a_2),p)$. So we can define
  $\J_AB\,\beta\in\E(\Gamma.A.A\,{\fst}.\Id_A\ent B)$ by:
  \begin{equation}
    \label{eq:25}
    \J_A B\,\beta\,(((x,a_1), a_2),p) \defeq
    \trpt{\left(\abs{i}{\shape{p}}(((x,a_1), p\at i),
        \upto{i}{p})\right)}\beta(x,a_1)  
  \end{equation}
  $\J$ has the required computation property, because
  \begin{align*}
    ((\J_A B\,\beta)\refl_A)(x,a) 
    &= \J_AB\,\beta\,(((x,a),a),\idp a)
    &&\text{by \eqref{eq:35}}\\
    &= \trpt{\left(\abs{i}{\src}{(((x,a),a)\idp a)}\right)}\beta(x,a)
    &&\text{by \eqref{eq:25} and \eqref{eq:36}}\\
    &= \trpt{\idp(((x,a),a)\idp a)}\beta(x,a)
    &&\text{by \eqref{eq:20}}\\
    &= \beta(x,a)
    &&\text{by \eqref{eq:10}}
  \end{align*}
  Stability of $\refl$ under re-indexing follows from the fact that
  the congruence operations $\gamma\pcong{}$ preserve degenerate
  paths; and stability of $\J$ uses \eqref{eq:19} and the fact that
  the path contraction operation $\upto{i}{(\_)}$ is preserved by the
  congruence operations $\gamma\pcong{}$.
\end{proof}

\begin{rem}[{Associative tap fibrations}]
  Definition~\ref{def:tapf} does not require the transport action
  $(\trpt{(\_)} : x\pth y \fun (A\,x \fun A\,y) \mid x,y:\Gamma)$ of
  a tap fibration $A\in\Fib(\Gamma)$ to be strictly associative. In
  other words, for all paths $p:x\pth y$ and $q:y\pth z$ and all
  $a:A\,x$ we do not necessarily have
  \begin{equation}
    \label{eq:42}
    \trpt{(q\pcomp p)} a = \trpt{q}(\trpt{p}a)
  \end{equation}
  As we have seen, that property is not needed to prove that
  fibrations give a model of type theory. Nevertheless, if one changes
  the definition by requiring \eqref{eq:42}, then the analogues of
  Theorems~\ref{thm:sptyp}--\ref{thm:wtyp} and \ref{thm:idet} do hold
  for this stronger notion of fibration, although we do not prove that
  here.

  With Definition~\ref{def:tapf} as it stands, one has associativity
  up to homotopy, i.e.~there is a path
  $\trpt{(q\pcomp p)} a \pth \trpt{q}(\trpt{p}a)$. This is a
  consequence of Theorem~\ref{thm:idet}, which allows one to use
  path induction~\cite[Section~2.9]{HoTT}) to construct a path
  $\trpt{(q\pcomp p)} a \pth \trpt{q}(\trpt{p}a)$ for any $p$ from
  the case when $p = \idp x$, where one has the degenerate path for
  $\trpt{(q\pcomp (\idp x))} a = \trpt{q}a = \trpt{q}(\trpt{(\idp
    x)}a)$.
\end{rem}

\begin{rem}[{Weak fibrations}]
  In the definition of tap fibration, if one replaces the equality
  \eqref{eq:10} by a homotopy, $\trpt{(\idp x)}a\pth a$, the resulting
  weak notion of tap fibration satisfies versions of
  Theorems~\ref{thm:sptyp}--\ref{thm:wtyp}, albeit with more
  complicated proofs. The same is true for Theorem~\ref{thm:idet}
  except that one gets only \emph{propositional} identity types, where
  there is a path between $(\J_A B\,\beta)\,{\refl_A}$ and $\beta$,
  rather than an equality in $\E$
  (cf.~\cite[Section~9.1]{CoquandT:cubttc} and
  \cite{VanDenBergB:patcpi}).
\end{rem}

\section{Function Extensionality}
\label{sec:fune}

Let $(\E,\R)$ be an order-ringed topos
(Definition~\ref{def:ordrt}). In this section we show that in the CwF
$\Fib$ constructed from $(\E,\R)$ as in the previous section, it is
the case that functions behave extensionally with respect to the
intensional identity types given by Moore paths.

So far we have only used the order and additive
structure of $\R$ (axioms \eqref{eq:38}--\eqref{eq:46} in
Fig.~\ref{fig:ordcra}). For function extensionality to hold, we need
more than that. To see why, consider the constant function
$\K_{\src}=\lambda i\fun \src:\I\fun\I$ and the identity function
$\id:\I\fun\I$.\footnote{Function extensionality for $\Fib$ only
  concerns functions between fibrant families; but as we noted in
  Section~\ref{sec:tap}, in $\Fib$ all objects, and in particular $\I$,
  are fibrant as trivial families over the terminal object.}  If
equality means existence of a Moore path, then these functions are
extensionally equal, because we have
$\lambda i\fun\abs{j}{i}j : \prod_{i:\I}(\K_{\src}\,i \pth \id\,i)$.
So if the principle of function extensionality is to hold with respect
to Moore paths, then there will have to be a path
$p: \K_{\src}\pth \id$ in $\I\fun\I$. What is its shape $\shape{p}$?
Since $p$ does not depend upon any assumptions, $\shape{p}$ would have
to be a global element $1\morphism\I$ in the topos $\E$. Axioms
\eqref{eq:38}--\eqref{eq:46} only guarantee the existence of one such
global element, namely $\src$. However, if $\shape{p}=\src$ then
$\K_{\src}=\id$, so $(\forall i:\R)\;\src=i$ and the model of type
theory is degenerate. Therefore we need some extra assumptions about
$\R$ if function extensionality is to hold without collapsing
everything. This is where axioms~\eqref{eq:43}--\eqref{eq:45} come
into play. These may not be the minimal assumptions needed for
function extensionality, but they are a well-known part of
(constructive, ordered) Algebra~\cite[chapter~VI]{BourbakiN:alg} that
does the job, and this makes easier the task of finding specific
models with good properties (which we address in
Section~\ref{sec:grotm}).

Given an object $\Gamma\in\E$ and a family $A\in\E(\Gamma)$, if
$p:f\pth g$ is a path between dependent functions
$f,g:\prod_{x:\Gamma}A\,x$, then for each $x:\Gamma$ we can apply the
path congruence operation \eqref{eq:7} to
$\lambda f\fun f\,x : (\prod_{x:\Gamma}A\,x) \fun A\,x$ and $p$ to
obtain a path $(\lambda f\fun f\,x)\pcong p : f\,x \pth g\,x$ in
$A\,x$. This gives us the following function (which coincides with the
canonical function obtained by path induction as
in~\cite[Section~2.9]{HoTT}):
\begin{align}
  &\begin{array}[c]{@{}l}
     {\happly}
     : \textstyle (f\pth g) \fun \prod_{x:\Gamma}(f\,x\pth g\,x)\\
     \happly\,p\,x \defeq (\lambda f\fun f\,x)\pcong p 
   \end{array}
  &&\text{(where $\textstyle f,g:\prod_{x:\Gamma}A\,x$)}
\end{align}

\begin{thm}[{function extensionality modulo $\pth$}]
  \label{thm:funext}
  Given an order-ringed topos $(\E,\R)$, an object $\Gamma\in\E$ and a
  family $A\in\E(\Gamma)$, for all $f,g:\prod_{x:\Gamma}A\,x$ there is
  a function
  $\funext:(\prod_{x:\Gamma}(f\,x\pth g\,x)) \fun (f\pth g)$.
  Furthermore this function is
  quasi-inverse~\cite[Definition~2.4.6]{HoTT} to $\happly$, that is,
  that for all $e : \prod_{x:\Gamma}(f\,x\pth g\,x)$ and $p:f\pth g$
  there are paths $\varepsilon\,e :\happly(\funext e) \pth e$ in
  $\prod_{x:\Gamma}(f\,x\pth g\,x)$ and
  $\eta\,p: p \pth \funext(\happly p)$ in $f\pth g$.
\end{thm}
\begin{proof}
  If $e: \prod_{x:\Gamma}(f\,x\pth g\,x)$, then for all $x:\Gamma$ we
  have $e\,x\at(\src\mult\shape{e\,x}) = e\,x\at\src = f\,x$ and
  $e\,x\at(\unit\mult\shape{e\,x}) = e\,x\at \shape{e\,x} = g\,x$. So there
  is a path $\funext e :f\pth g$ in $\prod_{x:\Gamma}A\,x$ given by:
  \begin{equation}
    \label{eq:27}
    \funext e = \abs{i}{\unit}\lambda x\fun e\,x\at(i\mult\shape{e\,x})  
  \end{equation}
  Note that by \eqref{eq:19} we have for each $x:\Gamma$
  \begin{align*}
    \happly(\funext e)\,x 
    &= (\lambda f\fun f\,x)\pcong
      \abs{i}{\unit}\lambda x \fun e\,x\at(i\mult\shape{e\,x})\\ 
    &= \abs{i}{\unit}(e\,x \at(i\mult\shape{e\,x}))
  \end{align*}
  whereas $e\,x = \abs{i}{\shape{e\,x}}(e\,x\at i)$ by \eqref{eq:21}. So
  to get a path from $\happly(\funext e)\,x$ to $e\,x$ we need to
  interpolate shapes from $\unit$ to $\shape{e\,x}$ while at the same
  time interpolating the argument of $e\,x\at\_$ from $i\mult\shape{e\,x}$
  to $i\mult\unit = i$. So for each $j:\R$ with $\src\tot j\tot\unit$,
  consider
  \[
    u_{j,x} \defeq (\unit - j) + j\mult\shape{e\,x} \quad\text{and}\quad
    v_{j,x} \defeq (\unit - j)\mult\shape{e\,x} + j
  \]
  Calculating with the ring axioms we find that
  $u_{j,x}\mult v_{j,x} =
  \shape{e\,x}+j\mult(\unit-j)\mult(\shape{e\,x}-\unit)^2$.
  Since $\src\tot j$ and $\src\tot \unit-j$ by assumption and since
  squares are always positive in an ordered ring
  ~\cite[VI.19]{BourbakiN:alg}, we have that
  $\shape{e\,x}\tot u_{j,x}\mult v_{j,x}$; hence
  $e\,x\at(u_{j,x}\mult v_{j,x}) = g\,x$. So for all $j:\R$ with
  $\src\tot j\tot\unit$ we have a path
  $\abs{i}{u_{j,x}}(e\,x\at(i\mult v_{j,x})) : f\,x \pth g\,x$ in
  $A\,x$.  When $j=\src$ this path is $\happly(\funext e)\,x$ (because
  $u_{\src,x} = \unit$ and $v_{\src,x} =\shape{e\,x}$); when $j=\unit$
  the path is $e\,x$ (because $u_{\unit,x} = \shape{e\,x}$ and
  $v_{\unit,x} =\unit$). Therefore we can define
 \begin{equation}
   \label{eq:28}
   \varepsilon\,e \defeq \abs{j}{\unit}\lambda x \fun \abs{i}{u_{j,x}} (e\,x
   \at(i\mult v_{j,x}))
 \end{equation}
 to get the desired path $\happly(\funext e) \pth e$ in
 $\prod_{x:\Gamma}(f\,x\pth g\,x)$.  Since for any $p:f\pth g$ it is
 the case that $p=\abs{i}{\shape{p}}(p\at i)$ and
 $\funext(\happly p) =\abs{i}{\unit}(p\at(i\mult\shape{p}))$, a similar
 argument to the one for $\varepsilon$ shows that
 \begin{equation}
   \label{eq:29}
   \eta\,p \defeq \abs{j}{\unit}\abs{i}{(\unit-j)\mult\shape{p}+j}(p\at
   (i\mult(\unit-j+j\mult\shape{p})))
 \end{equation}
 gives a path $p \pth \funext(\happly p)$ in $f\pth g$.
\end{proof}

\section{Universes}
\label{sec:uni}

Let $(\E,\R)$ be an order-ringed topos
(Definition~\ref{def:ordrt}). In this section we show how to construct
Tarski-style universes~\cite[p.~88]{Martin-LoefP:inttt} in the CwF
$\Fib$ constructed from $(\E,\R)$ as in Section~\ref{sec:tap}. To do
so we assume that (the CwF associated with) $\E$ supports
\emph{inductive-recursive} definitions~\cite{DybjerP:genfsi}. This
will be the case if $\E$ is a Grothendieck topos, that is, a category
of $\Set$-valued sheaves on a site, if we assume $\Set$ is a model of
a sufficiently strong set theory;
see~\cite[section~6]{DybjerP:finair}. As well as the topos $\E$, we
also need to assume something about the ordered ring $\R$, namely that
it is a connected object in $\E$; see the definition below. The
examples of order-ringed toposes that we give in Section~\ref{sec:mod}
have these properties. For simplicity, we confine attention to
universes containing a type of Booleans and closed under taking
dependent function and identity types (given by Moore paths); other
typing constructs can be dealt with in the same way.

Let the object $\UU\in\E$ and the family $\TT\in\E(\UU)$ be defined
simultaneously by induction-recursion in $\E$ so that $\UU$ has
inductive constructors
\begin{equation}
  \label{eq:49}
  \code{bool} : \UU
  \qquad
  \code{pi} : \textstyle\prod_{u:\UU}(\TT\,u\fun\UU)\fun\UU
  \qquad
  \code{eq} : \textstyle\prod_{u:\UU}\TT\,u\fun\TT\,u \fun \UU
\end{equation}
and $\TT$ satisfies the following recursion equations
\begin{align}
  &\TT\,\code{bool} = 1+1\\
  &(\forall u:\UU)(\forall f: \TT\,u\fun\UU)\; \TT(\code{pi}\,u\,f) = 
    \textstyle\prod_{x:\TT\,u}\TT(f\,x)\\
  &(\forall u:\UU)(\forall x,y:\TT\,u)\; \TT(\code{eq}\,u\,x\,y) =
    x\pth y\label{eq:50}
\end{align}
Here $x\pth y$ is the Moore path object (Definition~\ref{def:moopo})
and $1+1$ is the coproduct of the terminal object $1\in\E$ with
itself---this gives an object of Booleans in $\E$ with elements
\begin{equation}
  \label{eq:13}
  \true : 1+1 \qquad \false : 1+1
\end{equation}

Recall from the proof of Theorem~\ref{thm:empub} that any object of
$\E$, and in particular $\UU$, has a tap fibration structure when
regarded as a family over $1$. So the main task is to show that there
is a tap fibration structure for the family $\TT\in\E(\UU)$. The way
that we will construct this structure agrees with the tap fibration
structure for Boolean, $\Pi$-{} and identity types (as in
Theorems~\ref{thm:empub}, \ref{thm:sptyp} and \ref{thm:idet}
respectively) when it is re-indexed along the coding functions
$\code{bool}$, $\code{pi}$ and $\code{id}$. In this way we get a
Tarski-style universe $\TT\in\Fib(\UU)$ in $\Fib$ containing a type of
Booleans and closed under taking dependent function and identity
types.

Recalling the definition of a tap structure (Definition~\ref{def:tapf}),
given a Moore path $p:u_0\pth u_1$ in $\UU$ we wish to construct a
function $\trpt{p}:\TT\,u_0\fun\TT\,u_1$ which is the identity when
$p$ is degenerate. To do so we recurse over the structure of $u_0:\UU$,
which is of the form $\code{bool}$, or $\code{pi}\,v_0\,f_0$ (for some
unique $v_0$ and $f_0$), or $\code{eq}\,v_0\,x_0\,y_0$ (for some
unique $v_0$, $x_0$ and $y_0$).  In each case we would like the whole
path $p$ to remain within whichever constructor form $u_0$ has, so
that previously constructed transport functions can be combined
appropriately using the recipes in the proofs of
Theorems~\ref{thm:empub}, \ref{thm:sptyp} and \ref{thm:idet}. This
property of paths in $\UU$ holds if $\R$ is connected in the following
sense. 

\begin{defi}[{Connected objects in a topos}]
  \label{def:conot}
  For each object $X$ of the topos $\E$, consider the morphism
  $\K:1+1 \morphism (X \fun {1+1})$ from the Booleans $1+1$ into the
  exponential of $1+1$ by $X$ which is the transpose of the first
  projection $\pi_1:(1+1)\times X\morphism{1+1}$; thus $\K$ sends a
  Boolean $b:1+1$ to the constant function with value $b$. We will say
  that $X$ is \emph{connected} if $\K$ is an isomorphism.
\end{defi}

If $(\E,\R)$ is an order-ringed topos for which $\R$ is connected,
then given a predicate $\varphi:\I\fun\Omega$ that is decidable
($(\forall i:\I)\;\varphi\,i \disj \neg(\varphi\,i)$), consider the
function $f:\R\fun{1+1}$ well-defined by
\[
  f\,x =
  \begin{cases}
    \true
    &\text{if $({x\tot \src} \conj \varphi\,\src)
      \;\disj\; ({\src\tot x} \conj \varphi\,x)$}\\
    \false
    &\text{if $({x\tot \src} \conj \neg\varphi\,\src)
      \;\disj\; ({\src\tot x} \conj \neg\varphi\,x)$}
  \end{cases}
  \qquad(x:\R)
\]
Since $\R$ is connected, either $f=\K\,\true$, or $f=\K\,\false$. If
$\varphi\,\src$ holds, then $f\,\src=\true$ and therefore we cannot
have $f=\K\,\false$; so in this case we must have $f=\K\,\true$. Thus
$\R$ being connected implies that the half-line $\I$ satisfies
\begin{equation}
  \label{eq:15}
  (\forall \varphi:\I\fun\Omega)\;
  ((\forall i:\I)\;\varphi\,i \disj \neg(\varphi\,i)) \;\imp\;
  \varphi\,\src \imp (\forall i:\I)\; \varphi\,i
\end{equation}
We use this property to construct the required tap fibration structure
for $\TT$.

\begin{thm}[{Universes}]
  \label{thm:uni}
  Given an order-ringed topos $(\E,\R)$, 
  assuming the topos $\E$ supports the inductive-recursive definition
  \eqref{eq:49}--\eqref{eq:50} and that $\R$ is connected, then
  $\TT\in\E(\UU)$ has a tap fibration structure that make it into a
  Tarski-style universe in $\Fib$ containing a Boolean type
  and closed under dependent function and identity types. 
\end{thm}
\begin{proof}
  We can construct a transport function
  \begin{equation}
    \label{eq:26}
    (\forall u_0,u_1:\UU)(\forall p:u_0\pth u_1)(\forall a:\TT\,u_0)\;
    \trpt{p}a:\TT\,u_1 
  \end{equation}
  and prove
  \begin{equation}
    \label{eq:34}
    (\forall u:\UU)(\forall a:\TT\,u)\; \trpt{(\idp u)}a = a
  \end{equation}
  simultaneously by recursion and induction on the size (number of
  nested constructors) of elements of $\UU$.
  Given
  $(f,i):f\,\src\pth f\,i$ and $a:\TT(f\,\src)$ consider the structure
  of $f\,\src$:
  
  \paragraph{\textbf{Case $f\,\src = \code{bool}$}} Let
  $\varphi:\I\fun\Omega$ be $\varphi\,j\defeq (f j =
  \code{bool})$. Since $ \UU$ is the disjoint union of the images of the
  constructors $\code{bool}$, $\code{pi}$ and $\code{eq}$, we have that
  $\varphi$ is decidable; and it holds when $j=\src$. Therefore by
  \eqref{eq:15} we have $f\,i =\code{bool}$. So
  $a:\TT(f\,\src) = \TT\,\code{bool} = \TT(f\,i)$ and in this case we
  can define $\trpt{(f,i)}a \defeq a$.
  
  \paragraph{\textbf{Case $f\,\src = \code{pi}\,v_0\,g_0$, for some
      $v_0,g_0$}} Letting $\varphi:\I\fun\Omega$ be
  \[
    \varphi\,j\defeq (\exists v:\UU,g:\TT\,v\fun\UU)\; f\,j =
    \code{pi}\,v\,g
  \]
  once again this is a decidable predicate that holds at
  $\src$. Therefore by \eqref{eq:15} and the fact that $\code{pi}$ is
  injective, we have that there are functions $v:\I\fun\UU$ and
  $g:\prod_{j:\I}\TT(v\,j)\fun\UU$ with
  \[
    (\forall j:\I)\; f\,j =
    \code{pi}\,(v\,j)(g\,j)   \;\conj\; (v,i):v_0\pth v\,i \;\conj\;
    g\,\src=g_0 \;\conj\; (\forall j\totop i)\; g\,j=g\;i 
  \]
  Note that the type of $a$ is
  $\TT(f\,\src) = \prod_{x:\TT\,v_0}\TT(g_0\,x) =
  \prod_{x:\TT(v\,\src)}\TT(g\,\src\,x)$. To transport this to an
  element of $\TT(f\,i) = \prod_{x:\TT(v\,i)}\TT(g\,i\,x)$ we can use
  the construction \eqref{eq:17} from the proof of
  Theorem~\ref{thm:sptyp}. More explicitly, for each $x:\TT(v\,i)$ and
  $k:\I$ we have a path
  $\abs{j}{i\monus k}(v(i\monus j)) : v\,i\pth v\,k$ and by recursion we
  can form
  \[
    \overline{x}\,k \defeq \trpt{(\abs{j}{i\monus k}(v(i\monus j)))}x :
    \TT(v\,k)
  \]
  Hence we get $\overline{g}:\I\fun\UU$ and $a_0 :
  \TT(\overline{g}\,\src)$ by defining
  \begin{align*}
    \overline{g}\,k &\defeq g\,k\,(\overline{x}\,k)\\
    a_0 &\defeq a(\overline{x}\,\src)
  \end{align*}
  By the induction hypothesis $\overline{x}$ satisfies
  $(\forall k\totop i)\;\overline{x}\,k = x$ and hence
  $(\overline{g},i): \overline{g}\,\src \pth g\,i\,x$. So by recursion
  once again, we get $\trpt{(\overline{g},i)}a_0 :
  \TT(g\,i\,x)$; and by induction hypothesis, this is equal to $a\,x$ in
  the case $i=\src$. Abstracting over $x$ gives the required element
  $\trpt{(f,i)}a$ of $\prod_{x:\TT(v\,i)}\TT(g\,i\,x)$. 
  
  \paragraph{\textbf{Case $f\,\src = \code{eq}\,v_0\,x_0\,y_0$, for some
      $v_0,x_0,y_0$}} The argument is similar to the previous case, but
  using the decidable predicate
  $\varphi\,j \defeq (\exists v:\UU,x:\TT\,u, y:\TT\,u)\; f\,j =
  \code{eq}\,v\,x\,y$ and the construction \eqref{eq:48} from the proof
  of Lemma~\ref{lem:idet}.
\end{proof}

\begin{rem}[{Univalence}]
  An element $u$ of the universe $\UU$ is a code for the corresponding
  type $\TT\,u$. By the above theorem, every Moore path
  $p: u_0\pth u_1$ in $\UU$ induces a transport function
  $\trpt{p}:\TT\,u_0\fun\TT\,u_1$. Because Moore paths give identity
  types in $\Fib$, these functions are
  \emph{equivalences}~\cite[Chapter~4]{HoTT}, the inverse equivalence
  being given by transport along the reverse path $\rev p$. Were $\TT$
  to satisfy Voevodsky's \emph{univalence
    axiom}~\cite[Section~2.10]{HoTT}, every equivalence between
  $\TT\,u_0$ and $\TT\,u_1$ would have to be of the form $\trpt{p}$
  for some path $p: u_0\pth u_1$. This cannot be the case, because we
  have seen that connectedness of $\R$ implies that given
  $p: u_0\pth u_1$, the elements $u_0$ and $u_1$ must have equal
  outermost constructor form; whereas it is quite possible for $1+1$
  to be equivalent (indeed, isomorphic) to, for example, a $\Pi$-type
  named by a code in $\UU$.  So the universes constructed in this
  section do not satisfy the interesting form of extensionality
  embodied by the univalence axiom. We return to this point in the
  Conclusion.
\end{rem}

\section{Models}
\label{sec:mod}

In the previous sections we have seen how any order-ringed topos
$(\E,\R)$ (Definition~\ref{def:ordrt}) gives rise to a model of
Martin-L\"of type theory with intensional identity types given by
Moore paths on $\R$. If $\R$ is trivial, that is, satisfies
$\src=\unit$, then existence of a Moore path just coincides with
extensional equality in the topos $\E$. If the object $\R$ is
non-trivial, but has decidable equality in $\E$ (for example, if it is
the object of integers, or of rationals), then there is a path
$\abs{i}{\unit}(\mathtt{if\ } i=\src \mathtt{\ then\ } \true \mathtt{\
  else\ } \false) : \true \pth \false$ in the object $1+1$ of Booleans
and so in this case the model of type theory we get is logically
degenerate. Therefore, when searching for examples of order-ringed
toposes, we should at least look for ones with $\R$ non-trivial and
not decidable. We first give a simple example of such an order-ringed
topos, where the underlying topos is a presheaf category. Then we give
a more sophisticated example, using a sheaf topos and for which the
associated model of type theory has identity types that are not
necessarily truncated at any level; in other words, ones in which
iterated identity types $\Id_A,\Id_{\Id_A},\Id_{\Id_{\Id_A}},\ldots$
can be homotopically non-trivial for any level of iteration.

\subsection{A presheaf model}

Consider the category whose objects are ordered rings in the topos
$\Set$ (sets and functions) and whose morphisms are ordered ring
homomorphisms, that is, functions preserving the order $\tot$ and the
ring operations $\src,\unit,{+},{\mult}$. Let $\C$ be a small full
subcategory of this category. For the purposes of this example it is
not important which ordered rings $\C$ contains, so long as it
contains the terminal ordered ring $1$ (which has one element
$\src=\unit$) and a non-trivial ordered ring (one for which
$\src\not=\unit$); for definiteness let us assume that $\C$ contains
the reals $\RR$ with the usual order and ring operations.

\subsubsection{The topos}
\label{subsubsec:top}

We use the topos $\Set^\C$ of covariant presheaves on $\C$. Thus the
objects of $\Set^\C$ are functors $\C\morphism \Set$ and the morphisms
are natural transformations between such functors. Let
$\Delta:\Set\morphism\Set^\C$ denoted the functor assigning to each
set $S$ the constant presheaf $\Delta\,S:\C\morphism \Set$, whose
value at each $X\in\C$.  is $\Delta\,S(X) = S$. $\Delta$ is the
inverse image part of the unique geometric morphism from the topos
$\Set^\C$ to $\Set$; its direct image part, the right adjoint to
$\Delta$, is the global sections functor
$\Set^\C(1,\_):\Set^\C\morphism \Set$.  Below we will also need to use
the fact that $\Delta$ has a left adjoint
\begin{align}
  &\pi_0:\Set^\C\morphism \Set\\
  &\pi_0(F) \defeq F(1)\notag
\end{align}
(The adjunction $\pi_0\dashv \Delta$  follows from the fact that
$1$ is a terminal object in $\C$.)

\subsubsection{The ordered commutative ring object $\R$}

Let $\R:\C\morphism\Set$ denote the forgetful functor sending each
ordered ring in $\C$ to its underlying set. As an object of $\Set^\C$,
$\R$ has the structure of an ordered ring object:
\begin{itemize}
  
\item ${\tot}\mono\R\times\R$ is the sub-presheaf of $\R\times \R$
  whose value each object $X\in\C$ is the subset
  ${\tot}(X)\subseteq \R(X)\times \R(X) = X\times X$ given by the
  order on $X$:
  \begin{align}
    {\tot}(X) \defeq \{(x,y)\in X\times X \mid x\tot y\}
  \end{align}
  Note that these subsets do form a subpresheaf, because each morphism
  in $\C$ is in particular an order-preserving function. Furthermore,
  ${\tot}\mono\R\times\R$ is a total order \eqref{eq:38}--\eqref{eq:2}
  in the topos $\Set^\C$, because disjunction in a presheaf topos is
  computed component-wise and each ${\tot}(X)$ is a total order on
  $X$.
  
\item The ring structure on $\R$ is given component-wise by the ring
  structure on each $X\in\C$. For example, the addition morphism
  ${+}:\R\times \R \morphism \R$ has component at $X\in\C$ given by
  addition in $X$: ${+}_X(x,y) = x + y$; these are the components of a
  morphism in $\Set^\C$ because they are natural in $X$, since each
  morphism $\theta$ in $\C$ satisfies
  $\theta(x+y) = \theta\,x + \theta\,y$. Similarly, this structure on
  $\R$ satisfies axioms \eqref{eq:9}--\eqref{eq:45} because each
  $\R(X)= X$ is an ordered ring.
\end{itemize}

\subsubsection{$\R$ is connected}

To apply the results of Section~\ref{sec:uni}, we need to verify that
$\R$ is a connected object in $\Set^\C$
(Definition~\ref{def:conot}). Since the object of Booleans $1+1$ in
the presheaf topos $\Set^\C$ is isomorphic to the constant functor
$\Delta\,2:\C\morphism\Set$, where $2=\{0,1\}$ is a two-element set,
we have to show that $\K:\Delta\,2 \morphism (\R\fun\Delta\,2)$ is an
isomorphism in $\Set^\C$. To see this, we just need to show that
$\K_X:\Delta\,2(X)\morphism(\R\fun\Delta\,2)(X)$ is a bijection for
each $X\in\C$. Letting $\Y:\C^\op\morphism \Set^\C$ denote the Yoneda
embedding and using the adjunction $\pi_0\dashv\Delta$ mentioned in
Sec.~\ref{subsubsec:top}, there are bijections
\begin{align*}
  \Delta\,2(X)
  &= 2
  &&\text{(definition of $\Delta$)}\\
  &\bij\Set(\C(X,1)\times 1, 2)
  &&\text{(since $1$ is terminal in $\C$)}\\
  &= \Set((\Y X\times \R)(1),2)
  &&\text{(definition of $\Y X\times \R$)}\\ 
  &= \Set(\pi_0(\Y X\times \R),2)
  &&\text{(definition of $\pi_0$)}\\
  &\bij \Set^\C(\Y X\times \R, \Delta\,2)
  &&\text{($\pi_0$ left adjoint to $\Delta$)}\\
  &\bij \Set^\C(\Y X,\R\fun\Delta\,2)
  &&\text{(universal property of exponential)}\\
  &\bij (\R\fun\Delta\,2)(X)
  &&\text{(Yoneda Lemma~\cite[III.2]{MacLaneS:catwm})}
\end{align*}
and one can check that their composition is $\K_X$.

So $(\Set^\C,\R)$ is an order-ringed topos from which we can construct
a model $\Fib$ of intensional Martin-L\"of Type Theory with universes
as in the previous sections. Since we assumed that $\C$ contains the
non-trivial ordered ring $\RR$, it is not the case that $\R$ is
trivial; that is, the sentence $\src=\unit$ is not satisfied by $\R$
(because it is not satisfied at component $X=\RR$).\footnote{On the
  other hand, since we assume $\C$ also contains the trivial ring $1$,
  neither is it the case that the sentence $\neg(\src=\unit)$ is
  satisfied by $\R$. Note that $\Set^\C$ is not a Boolean topos -- it
  does not satisfy the Law of Excluded Middle.} More than this, the
associated model of Type Theory built from $(\Set^\C,\R)$ as in the
previous sections is not logically trivial: there is no proof in it of
$\Id_{\Bool}\true\,\false$ (so in particular, $\R$ cannot be a
decidable object of $\Set^\C$). To see this, note that giving a proof
of $\Id_{\Bool}\true\,\false$ in the CwF $\Fib$ associated with
$(\Set^\C,\R)$ is the same as giving an $\R$-based Moore path from
$\true:1\morphism \Delta\,2$ to $\false:1\morphism \Delta\,2$. There
is no such path because we saw above that
${\R\fun\Delta\,2} \bij \Delta\,2$ and hence every path in $\Delta\,2$
is constant.

\subsection{A gros topos model}
\label{sec:grotm}

Being logically non-trivial is rather a weak condition for models of
type theory with intensional identity types. A more interesting one is
that the model contains types $A$ whose iterated identity types
$\Id_A,\Id_{\Id_A},\Id_{\Id_{\Id_A}},\ldots$ are all non-trivial (not
isomorphic to the unit type).  An interesting way of demonstrating
that for the model associated with an order-ringed topos $(\E,\R)$ is
to show that the homotopy types of a rich collection of topological
spaces (including all the $n$-dimensional spheres, say) are faithfully
represented by the internal, $\R$-based notion of homotopy on a
corresponding collection of objects of the topos $\E$. We give one
such order-ringed topos in this section.

\subsubsection{The topos}

We use a topos of sheaves $\Sh(\T,\Cov)$~\cite[III]{MacLaneS:shegl}
for the following site $(\T,\Cov)$. Let $\Top$ denote the category of
Hausdorff topological spaces and continuous functions. We take the
small category $\T$ to be the least full subcategory of $\Top$
containing the reals $\RR$ and closed under the following operations:
\begin{itemize}

\item[(a)] If $X,Y\in\T$, then $\T$ contains the product space
  $X\times Y$.

\item[(b)] If $X\in\T$ and $C\subseteq X$ is a closed subset, then $C$
  (with the subspace topology) is in $\T$.

\item[(c)] If $X,Y\in\T$ and $X$ is locally compact, then $\T$
  contains the exponential space $Y^X$ (the set of continuous
  functions from $X$ to $Y$ endowed with the compact-open topology).

\end{itemize}
Since the spaces in $\T$ are Hausdorff, equalizers of continuous
functions give closed subspaces and hence by (a) and (b) we have that
$\T$ is closed under taking finite limits in $\Top$. Hence $\T$
contains $\pos{\RR}$ (as well as $[0,1]$ and all the spheres $S^n$ for
any $n\in\N$). Then by (c) we have that $\T$ is closed under taking
Moore path spaces (with their usual topology):
\begin{equation}
  \label{eq:30}
  X\in\T \;\imp\; MX \defeq \{(f,r)\in
  X^{\pos{\RR}}\times\pos{\RR} \mid (\forall 
  r'\geq r)\; f\,r' = f\,r\} \in \T
\end{equation}
We define a coverage~\cite[Definition~A2.1.9]{JohnstonePT:skeett} on
$\T$ as follows.  Following Dyer and Eilenberg~\cite{DyerE:glofs} we
say that a set $S$ of closed subsets of a topological space $X$ is a
\emph{local cover} if for each $x\in X$ there is a finite subset
$\{C_1,\ldots,C_n\}\subseteq S$ with $x$ in the interior of
$C_1\cup\cdots\cup C_n$. Of course every finite cover of $X$ by closed
subsets is trivially a local cover. Note that if $f:Y\morphism X$ in
$\T$ and $S$ is a local cover of $X$, then
$f^{-1}S \defeq \{f^{-1}C \mid C\in S\}$ is a local cover of
$Y$. Hence we get a coverage on $\T$ with
$\Cov(X) \defeq \{S\mid \text{$S$ is a local cover of $X$}\}$ for each
$X\in\T$.

Given a functor $F:\T^\op\morphism \Set$, if $C\subseteq X$ is a
closed subspace of a space in $\T$, then for each $x\in F\,X$ we
will just write $x|_C$ for the element $F\,i\,x \in F\,C$,
where the $\T$-morphism $i:C\morphism X$ is the inclusion function.
Recall that $F$ is a sheaf with respect to $\Cov$ iff for all
$X\in\T$, $S\in\Cov(X)$ and all $S$-indexed families
$(x_C\in F\,C \mid C\in S)$ satisfying
$(\forall C,D\in S)\; x_C|_{C\cap D} = x_D|_{C\cap D}$, there is a
unique $x\in F\,X$ with $(\forall C\in S)\; x_C = x|_C$. The topos
$\Sh(\T,\Cov)$ is by definition the full subcategory of the functor
category $\Set^{\T^\op}$ whose objects are sheaves.

\subsubsection{The ordered commutative ring object $\R$}

Let $\Y:\T\morphism\Set^{\T^\op}$ denote the Yoneda embedding for the
small category $\T$. Because elements of $\Cov(X)$ are local covers of
$X\in\T$, for each $S\in\Cov(X)$ if we have a family of continuous
functions $f_C:C\morphism X'$ for each $C\in S$ that agree where they
overlap ($f_C|_{C\cap D}= f_D|_{C\cap D}$), then the unique function
$f:X\morphism X'$ that agrees with each of them ($f_C= f|_C$) is
necessarily continuous. Hence each representable presheaf
$\Y\,X'=\T(\_\mathbin{,}X')$ is a sheaf (in such a case one says that
the coverage $\Cov$ is \emph{sub-canonical}) and the Yoneda embedding
gives a functor $\Y:\T\morphism\Sh(\T,\Cov)$.

The axioms for \emph{partially} ordered commutative rings are those in
Fig.~\ref{fig:ordcra} except for \eqref{eq:2}. They make sense in any
category with finite limits once we have an interpretation of the
binary operation $\tot$, the constants $\src,\unit$ and the operations
${+},{\minus},{\,\mult\,}$; and satisfaction of those axioms is
preserved by functors that preserve finite limits.  Since
$\{(x,y)\in \RR\times\RR\mid x\leq y\}$ is a closed subset of
$\RR\times\RR$ and the usual ring operations on $\RR$ are continuous,
it follows that $\RR$ is a partially ordered commutative ring object
in the finitely complete category $\T$. Then since $\Y$ preserves
finite limits, the representable sheaf $\R \defeq \Y\RR$ is a
partially ordered commutative ring object in the topos
$\Sh(\T,\Cov)$. In fact it is totally ordered, that is, satisfies
\eqref{eq:2}. This is simply because we have a local cover
$\{\{(x,y)\mid x\leq y\},\{(x,y)\mid y\leq x\}\}
\in\Cov(\RR\times\RR)$.\footnote{This fact is the motivation for
  considering a site based on covers by closed subsets rather than the
  more familiar case of open covers used for Giraud's gros
  topos~\cite[IV, 2.5]{SGA4}. However, as Spitters has pointed out
  [private communication], in view of
  \cite[Theorem~8.1]{JohnstonePT:ontopt} we could have used the reals
  in Johnstone's \emph{topological topos} instead of $\Sh(\T,\Cov)$ in
  this section.}

\subsubsection{$\R$ is connected}

To apply the results of Section~\ref{sec:uni}, we need to verify that
$\R$ is a connected object in $\Sh(\T,\Cov)$
(Definition~\ref{def:conot}). First note that $\T$ contains the
discrete spaces $\emptyset$, $1$ and $2=\{0,1\}$. Since $\emptyset$ is
covered by the empty family of closed sets and $2$ is covered by the
two closed inclusions $\{0\}\hookrightarrow 2 \hookleftarrow\{1\}$,
for any sheaf $F\in\Sh(\T,\Cov)$ we have that $F(\emptyset)\bij 1$ and
then $F(2)\bij F(1)\times F(1)$. So by the Yoneda Lemma
\[
  \Sh(\T,\Cov)(\Y 2,F) \bij F(2) \bij F(1)\times F(1)
  \bij\Sh(\T,\Cov)(\Y1+\Y1,F)\bij \Sh(\T,\Cov)(1+1,F)
\]
naturally in $F$. So the object of Booleans in $\Sh(\T,\Cov)$ is
representable by $2\in\T$. Thus to see that $\R$ is connected, it
suffices to show that $\K:\Y 2 \morphism (\R\fun\Y 2)$ is an
isomorphism in $\Sh(\T,\Cov)$. But $\R=\Y\RR$ is also representable
and $\Y$ preserves exponentials. It follows that
$\K:\Y 2 \morphism (\R\fun\Y 2)$ is the image under $\Y$ of the
continuous function $2\morphism 2^\RR$ given by taking the exponential
transpose of $\pi_1:2\times \RR \morphism 2$ in $\T$. But
$2\morphism 2^\RR$ is an isomorphism in $\T$ because $\RR$ is
topologically connected; and hence so is
$\K:\Y 2 \morphism (\R\fun\Y 2)$ in $\Sh(\T,\Cov)$.

\subsubsection{Homotopy types in $\Sh(\T,\Cov)$}

From the order-ringed topos $(\Sh(\T,\Cov), \Y\RR)$ we get a CwF
$\Fib$ modelling Martin-L\"of Type Theory with universes and with
intensional identity types $\Id_A$ given by Moore paths.  Consider the
congruence on the category $\Sh(\T,\Cov)$ that identifies morphisms
$\gamma,\delta:B\morphism A$ when there is a global section of
$\Id_{B\fun A}\gamma\,\delta$. Let $\Ho(\Sh(\T,\Cov))$ denote the
associated quotient category.

By Theorem~\ref{thm:funext}, two
morphisms $\gamma,\delta:B\morphism A$ are identified in
$\Ho(\Sh(\T,\Cov))$ iff $\prod_{x:B}(\gamma\,x\pth \delta\,x)$ has a
global section. This is equivalent to requiring the existence of a
morphism $H:B\morphism\Path A$ in $\Sh(\T,\Cov)$ satisfying
$\source\comp H = \gamma$ and $\target\comp H = \delta$, where
\begin{gather}
  \label{eq:24}
  \Path A \defeq \{(f,i):(\I\fun A)\times\I \mid (\forall
  j\totop i)\; f\,j= f\,i\}\\
  \source, \target : \Path A\morphism A
  \qquad \source(f,i) \defeq f\,\src 
  \qquad \target(f,i) \defeq f\,i \notag
\end{gather}
is the total object of the family of Moore path objects \eqref{eq:3}
with associated source and target morphisms. Note that since
$\{(x,y)\in\RR\times\RR \mid x\leq y\}$ is a closed subset of
$\RR\times\RR$, it is an object of $\T$ and the order relation on $\R$
in $\Sh(\T,\Cov)$ is given by the representable
$\Y \{(x,y)\mid x\leq y\}\mono \Y(\RR\times\RR)\bij \R\times \R$.  It
follows that the positive cone of $\R$ is also a representable sheaf:
$\I\iso\Y(\pos{\RR})$. Recall that, as well as preserving finite
limits, the Yoneda embedding preserves any exponentials that happen to
exist. It follows that for each representable sheaf $\Y X$ its total
path object is representable by the Moore path space of
$X$~\eqref{eq:30}: $\Path(\Y X) \iso \Y(M X)$; and under this
isomorphism the source and target morphisms $\source,\target$
correspond to the representable morphisms induced by the usual source
and target functions for $M X$.  Since $\Y$ is full and faithful, it
follows that for two continuous functions $f,g:X\morphism X'$ in $\T$,
the morphisms $\Y f$ and $\Y g$ get identified in $\Ho(\Sh(\T,\Cov))$
iff $f$ and $g$ are homotopic in the classical sense. Thus $\Y$
induces a full and faithful embedding of the category $\Ho(\T)$ of
homotopy types of spaces in $\T$ into $\Ho(\Sh(\T,\Cov))$. Since $\T$
contains all the spheres $S^n$, we deduce that identity types in this
particular $\Fib$ are not necessarily truncated at any level of
iteration.

\section{Related Work}
\label{sec:relw}

The classical topological notion of Moore path is a standard, if
somewhat niche topic within homotopy theory. The Schedule Theorem of
Dyer and Eilenberg~\cite{DyerE:glofs} for globalising Hurewicz
fibrations is a nice example of their usefulness. They have been used
in connection with higher-dimensional category theory by Kapranov and
Voevodsky~\cite{VoevodskyV:infght} and by
Brown~\cite{BrownR:moohsf}. 

Although our use of constructive algebra within toposes to make models
of intensional type theory appears to be new, we are not the only ones
to consider using some form of path with strictly unitary and
associative composition to model identity types with a judgemental
computation rule.  Van~den~Berg and Garner~\cite{VanDenBergB:topsmi}
use topological Moore paths and a simplicial version of them to get
instances of their notion of \emph{path object category} for modelling
identity types.  The results of Sections~\ref{sec:moopt} and
\ref{sec:tap} show that any ordered abelian group in a topos induces a
path object category structure on that topos; and since the notion of
fibration we use (Definition~\ref{def:tapf}) is closely related to the
one used in \cite{VanDenBergB:topsmi} (see Proposition~6.1.5 of that
paper), one can get alternative, more abstract categorical proofs of
Theorems~\ref{thm:sptyp} and \ref{thm:idet} from the work of
Van~den~Berg and Garner. However, the concrete calculations in the
internal language that we give are quite simple by comparison; and
this approach proves its worth in Section~\ref{sec:fune}, whose
results on obtaining function extensionality from the ordered ring
structure are new.

The PhD thesis of North~\cite{NorthPR:typtwf} uses a
category-theoretic abstraction of the notion of Moore paths, called
\emph{Moore relation systems}, as part of a complete analysis of when
a weak factorization system gives a model (in terms of display map
categories, rather than CwFs) of identity-{}, $\Sigma$-{} and
$\Pi$-types. A Moore relation system is a piece of category theory
comparable to our use of ordered abelian groups in categorical logic
in Section~\ref{sec:moopt}; it would be interesting to see if it can
be extended in the way we extended from groups to rings in
Section~\ref{sec:fune} in order to validate function extensionality.

Spitters~\cite[Section~3]{SpittersB:cubstt} uses a somewhat
different formulation of Moore path in the cubical
topos~\cite{CoquandT:cubttc}. His notion is the reflexive-transitive
closure of the usual path types given by the bounded interval. For
better properties and to get a closer relationship with our version,
one would like to quotient these cubical ``Spitters-Moore'' path
objects up to degenerate paths; but the undecidability of degeneracy
seems to stop one being able to do that while retaining the (uniform)
Kan-fibrancy of such path objects. Here we can side-step such issues,
since notably our models manage to avoid using a notion of Kan
fibrancy at all.

\section{Conclusion}
\label{sec:con}

We have shown that any connected ordered ring in a topos gives rise to
a model of Martin-L\"of's intensional Type Theory with universes in
which proofs of identity are given by Moore paths. We gave an example
to show that such models of type theory can contain highly non-trivial
identity types that faithfully represent the homotopy types of a wide
class of topological spaces and in particular are not truncated at any
level of iteration.

It is an open question whether there is an order-ringed topos that
gives rise to such a model of type theory containing a
\emph{univalent}~\cite[Section~2.10]{HoTT} universe. (We saw in
Section~\ref{sec:uni} that the universes constructed there are not
univalent.) The known examples of non-truncated univalent universes,
such as the classical simplicial sets model~\cite{LumsdainePL:simmuf}
and the various constructive cubical sets
models~\cite{CoquandT:modttc,CoquandT:cubttc,LicataDR:carctt} make use
of a modified form of the Hofmann-Streicher~\cite{HofmannM:lifgu}
construction in presheaf categories. Streicher~\cite{StreicherT:unit}
points out that the basic Hofmann-Streicher universe construction
works for sheaf toposes through a suitable use of sheafification. So
there are Hofmann-Streicher universes in both of the example toposes
from Section~\ref{sec:mod}, one of which is a presheaf topos and one a
sheaf topos. However, the analysis of \cite[Section~5]{PittsAM:intumh}
shows that in the examples of univalence mentioned above, one gets
from the Hofmann-Streicher universe to a univalent universe
classifying fibrations (of various kinds) by using the fact that in
those models the path functor $\Path(\_)$ has a right
adjoint. Unfortunately this is not the case for the models in
Section~\ref{sec:mod}, where it seems that the very property of the
interval (half-line) that allows us to avoid all uses of Kan filling
in favour of path composition when building our models of type theory,
namely the total order~\eqref{eq:1}, prevents the interval from being
``tiny'' and hence prevents $\Path(\_)$ from having a right adjoint.

\section*{Acknowledgement}

\noindent We are very grateful to Benno van~den~Berg and Bas Spitters
for discussions about the material in this paper. Orton was supported
by a PhD studentship from the UK EPSRC funded by grants EP/L504920/1
and EP/M506485/1.

\newcommand{\etalchar}[1]{$^{#1}$}

\end{document}